\documentclass[3p]{elsarticle}

\usepackage{geometry} 
\usepackage{amsfonts}
\usepackage{mathtools}
\usepackage{amsthm}
\usepackage{mathdots}
\usepackage{bbm}
\usepackage{hyperref} 
\usepackage{cleveref}
\usepackage{color}
\usepackage{subcaption}
\usepackage{natbib}
\usepackage{multirow}
\setcitestyle{authoryear, open={[},close={]}}
\usepackage[linesnumbered,ruled,vlined]{algorithm2e}

\usepackage{tikz}
\usetikzlibrary{arrows.meta}
\usetikzlibrary{arrows.meta, decorations.pathmorphing, fit}

\usepackage{array}
\usepackage{booktabs} 

\usepackage[symbol]{footmisc}

\input{Commands.tex}

\title{Modeling and Replication of the Prepayment Option of Mortgages including Behavioral Uncertainty\texorpdfstring{\footnote[2]{The views expressed in this paper are the personal views of the authors and do not necessarily reflect the views or policies of their current or past employers. The authors have no competing interests.}}{}}
\author[1]{Leonardo Perotti\texorpdfstring{\corref{cor1}}{}}
\ead{L.Perotti@uu.nl}
\author[1,2]{Lech A.~Grzelak}
\ead{L.A.Grzelak@uu.nl}
\author[1]{Cornelis W.~Oosterlee}
\ead{C.W.Oosterlee@uu.nl}
\cortext[cor1]{Corresponding author.}
\address[1]{Mathematical Institute, Utrecht University, Utrecht, the Netherlands}
\address[2]{Financial Engineering, Rabobank, Utrecht, the Netherlands}
\date{\today}

\begin{document}

\begin{abstract}
    \noindent Prepayment risk embedded in fixed-rate mortgages forms a significant fraction of a financial institution's exposure, and it receives particular attention because of the magnitude of the underlying market. The \emph{embedded prepayment option} (EPO) bears the same interest rate risk as an exotic interest rate swap (IRS) with a suitable stochastic notional. We investigate the effect of relaxing the assumption of a deterministic relationship between the market interest rate incentive and the prepayment rate \citep[see, e.g.,][]{casamassima2022pricing}. A non-hedgeable risk factor is modeled to capture the uncertainty in mortgage owners' behavior, leading to an incomplete market. 
    As claimed in \citep{bissiri2014modeling}, we prove under natural assumptions that including behavioral uncertainty reduces the exposure's value.
    We statically replicate the exposure resulting from the EPO with IRSs and swaptions, and we show that a replication based on swaps solely cannot easily control the right tail of the exposure distribution, while including swaptions enables that. The replication framework is flexible and focuses on different regions in the exposure distribution. 
    Since a non-hedgeable risk factor entails the existence of multiple equivalent martingale measures, pricing and optimal replication are not unique. We investigate the effect of a market price of risk misspecification and we provide a methodology to generate robust hedging strategies. Such strategies, obtained as solutions to a saddle-point problem, allow us to bound the exposure against a misspecification of the pricing measure.
\end{abstract}


\maketitle

\section{Introduction}
\label{sec: Intro}

From a risk management perspective, financial institutions have no particular interest in issuing fixed-interest rate mortgages. In fact, since a bank's funding is indexed to a reference floating rate, exposure to fixed rates gives rise to \emph{interest rate risk}. 
A bank can, in principle, completely remove all the interest rate risk by investing in a suitable combination of \emph{vanilla} interest rate swaps (IRSs). 
However, any variation from contractual cash flow payments results in a mismatch between the exposure of the mortgage portfolio and the portfolio of IRSs -- in the opposite direction. The phenomenon of \emph{mortgage prepayment} is one reason for such a mismatch. By the term prepayment, we mean all the mortgage notional repayments that occur and are not contractual, hence they are unknown a priori. The focus is on \emph{penalty-free} prepayments, i.e. prepayments that are not contractually embedded with the payment of a penalty to compensate the financial institution for the expected loss caused by the prepayment event. 
The risk connected to these occurrences is called in a broad sense \emph{prepayment risk}. 

Penalty-free prepayment allowances depend on contractual specifications and generally are governed by diverse laws and regulations, contingent upon specific regions or countries. For instance, in the US market, mortgage owners possess the right to prepay any desired amount without incurring penalties. In Canada, the most common threshold for penalty-free prepayment is paying 20\% of the initial notional per year, while in Spain -- even if fixed-rate mortgages are not popular -- the law sets the threshold to 10\% \citep{green2014introduction}. The Dutch market -- where this study is conducted -- is one of the major mortgage markets in the world (relative to their GDP). Most contracts have a medium-term fixing period, usually ranging between 5 and 10 years \citep{alink2002mortgage}, and the penalty-free yearly thresholds have franchises commonly set at 10\% or 20\% of the initial outstanding debt. Specific circumstances, like relocation, permit unrestricted prepayment without penalties.
Independently of the reason for prepayment, the magnitude of the Dutch mortgage market exposure is non-negligible.
Hence, a precise evaluation and assessment of the risk connected with the \emph{embedded prepayment option}\footnote{By ``embedded option'' we want to underlie that such \emph{optionality} is not explicitly written but arises from standard fixed-rate mortgage contractual features.} (EPO) in fixed-rate mortgages is crucial.

The prepayment risk encompasses two primary facets: liquidity risk, linked to mortgage funding costs, and interest rate risk, which is typically mitigated through (Delta-)hedging. We will focus on the assessment and static hedging of the latter. 
Fixed interest rates are coupled with IRSs to swap fixed for floating inflows. In particular, such a coupling involves payer (amortizing) IRSs featuring contractual characteristics akin to the mortgage portfolio, regarding the notional amortization scheme, fixed rate, payment frequency, and tenor. As a consequence, the net value of the acquired position remains virtually unaffected by underlying interest rate fluctuations, as the position is \emph{Delta-neutral}. However, a prepayment event introduces a disparity between the expected cash flows from the mortgage owner and the cash flows from the IRSs used for the replication of the original exposure. Consequently, the mortgage issuer's position becomes sensitive to the direction of the underlying market risk factor, thereby exposing it to interest rate risk.

The conventional academic approach employed to model prepayment comprehends two main classes: \emph{financially rational models} and \emph{exogenous models}. Financially rational models assume prepayment is only affected by market risk factors and actual prepayment only occurs when the decision of prepaying bears higher value -- in expectation -- than continuing without prepayment. For instance, \citep{kuijpers2007optimal} price prepayment for interest-only mortgages in the Dutch market as a multi-exercise Bermudan-type option. 
It is worth noticing that it emerges clear from the data that people do not act rationally when a prepayment decision is taken. 

Extensive empirical studies have been done to investigate the most relevant features explaining the prepayment phenomenon \citep[see, among others,][]{alink2002mortgage,charlier2003prepayment,kalotay2004option,hoda2007implementation,lin2010determinants,hassink2011importance,caspari2018valuation}. It has been shown that prepayment is subject to time effects (e.g. people in certain age buckets are more likely to prepay than people in other buckets), seasonality (e.g., due to taxation benefits, December and January show high rate of prepayment), and other sector-specific features (e.g., prepayment for relocation is affected by the housing market activity). The combination of these \emph{non-financial} features leads to a non-(financially) rational exercise of the EPO, which justifies the development of exogenous models. 

Nonetheless, all empirical studies agree on one fact: prepayment behavior of mortgage owners is mainly driven by an \emph{interest rate incentive} (RI), commonly represented as a spread or ratio between the contractual and the prevailing mortgage rate. Such a relationship is delineated through a monotonic function, often a \emph{sigmoid}, where the steepness characterizes the client's response rate to an enticing rate incentive. Distinct functions represent varying degrees of rationality exhibited by mortgage owners. The two extreme cases are the completely unaware owners and the fully rational owners (displaying a step function ranging from 0 to a maximum prepayment amount). At the portfolio level, intermediate functions encapsulate the ``average'' level of rationality among the mortgage owners in the portfolio. These functions are estimated using historical data and subsequently employed for valuation and hedging purposes. We refer, among others, to \citep{sherris1993pricing,hoda2007implementation,castagna2013measuring,casamassima2022pricing}.

However, this approach presents two primary limitations. Firstly, reliable hedging outcomes can only be achieved when the historical data adequately represents the mortgage owner's response to the rate incentive. Secondly, the approach fails to account for the random nature of prepayment events, namely the ``behavioral risk'' of \citep{bissiri2014modeling}.\footnote{For \citep{bissiri2014modeling}, behavioral risk includes all the decisions that are not driven by a fully financially rational criterion.} 
The objective of this paper is to investigate the impact of a stochastic factor driving the prepayment effect. By introducing a random factor, independent of the market risk factor, a robust representation of the prepayment event can be attained. This stochastic factor serves as a model to quantify the intrinsic \emph{behavioral uncertainty} inherent in prepayment, rendering the resulting model flexible and less reliant on historical data. Furthermore, a hedging strategy based on this model exhibits robustness.

From a risk management perspective, the problem of replicating a target exposure can be recast as an optimization task where the objective function is a metric of the ``path-wise distance'' of the target exposure (in our case the EPO value) to a certain portfolio of hedging/replicating instruments. By path-wise distance, we mean the difference between the exposure and the replicating portfolio values per market scenario. In the financial jargon, scenarios are often referred to as ``paths.'' The path-wise distance is observed over a time window of interest.
This approach, in contrast to conventional Delta-hedging techniques, offers the advantage of a more targeted focus on specific regions of the exposure distribution. By defining a suitable metric of interest, we can tailor the hedging strategy to closely replicate the behavior of the target exposure across different market scenarios and time horizons.
For instance, a trader in a financial institution might be interested in hedging the ``tail-risk'' rising from some rare, extreme movement in the market while being indifferent towards many small movements leading to limited exposure. Even more, she may be worried about the movements of her book's PnL in a specific direction.
Focus on the exposure tails is difficult to enforce within the framework of Delta-hedging but can be imposed by considering quantities such as the \emph{expected shortfall}.

\subsection{Contribution}
The goals of this research are twofold. Firstly, a relaxation of the assumption of deterministic mapping from a market risk factor to the realized prepayment is achieved by introducing a stochastic non-hedgeable risk factor for prepayment whose estimation is based on the real-world prepayment. We show theoretically (for a special case) and with numerical experiments that our model choice leads to low cost of prepayment.
Secondly, the pricing model is the input in a replication task via IRSs and swaptions. According to the assumptions on the equivalent martingale measure employed for pricing, two formulations of the (static) hedging problem are proposed, resulting in replicating strategies susceptible to model misspecification. When the pricing measure is assumed to be known, we can compute the optimal replicating strategy. Conversely, when such assumption is relaxed a robust optimal strategy is obtained, that guarantees a satisfactory performance against a misspecification of the pricing measure. The replication is based on path-wise matching of the EPO value. The approach is flexible and the inclusion of features of interest (such as focusing on a certain region of the exposure distribution) is straightforward by directly modifying the objective of the optimization.

The remainder of the paper is organized as follows. \Cref{sec: MortgagePrepayment} fixes the notation and introduces the background about fixed-rate mortgage prepayment. In \Cref{sec: PrepaymentModel} is presented the prepayment model. \Cref{ssec: PerturbedRIFunction,ssec: RiskNeutralEvaluationEPO,ssec: LowPriceEPO} focus on the inclusion and effect of the stochastic behavioral risk factor in the pricing model; while \Cref{ssec: HedgingPrepayment} is dedicated to the development of the static replication strategy. \Cref{sec: NumericalPricingReplication} presents the numerical methods used for pricing the EPO (\Cref{ssec: NumericalPricing}) and for its replication using market tradable instruments (\Cref{ssec: NumericalReplication}). \Cref{sec: NumericalResults} provides a detailed description of the conducted numerical experiments for what concern valuation (\Cref{sssec: CheapPriceEPO}) and replication (\Cref{ssec: ExposureReplicationSimple,ssec: ExposureReplicationRobust}) of the EPO. Finally, \Cref{sec: Conclusion} furnishes a comprehensive conclusion, summarizing the key findings.

\section{Mortgage prepayment}
\label{sec: MortgagePrepayment}

The current section provides the relevant definitions and notation regarding fixed-rate mortgages and the phenomenon of prepayment.
In the following, we will use the terms \emph{issuer} and \emph{owner} to indicate the financial institution (often a bank) that issues the mortgage and its counterparty, respectively.

\subsection{Fixed-rate mortgages}
\label{ssec: FixedRateMortgages}

From a contractual perspective, a fixed-rate mortgage is identified by a set of constituting features that determine the notional repayment and the interest payment schedules. 
\begin{dfn}[Mortgage features]
\label{def: MortgageSpecs}
Given the time horizon $\T=[0,T]$, where $T$ is the maximum observation time, let us define the mortgage \emph{issue date} $t_0\in\T$ and the mortgage \emph{payment dates}:
\begin{equation}
\label{eq: PaymentDates}
    \T_{\tt{p}}=\{t_1,\dots,t_n\}\in\T,\quad t_0 <t_1 <t_2<\dots <t_n,
\end{equation}
for $n>0$, the number of payments. We denote by $K\in\R$ the \emph{annualized fixed rate} used for the computation of the interest payments. We define the \emph{contractual notional} $N_{\tt{c}}(t)$, $t\in\T$, i.e. the amount of outstanding notional at time $t$, and we indicate $N_{\tt{c},0}=N_{\tt{c}}(t_0)$. We assume $N_{\tt{c}}(t)$ right-continuous, and we have $N_{\tt{c}}(t)=0$ for $t\geq t_n$.
\end{dfn}

\begin{rem}[Amortization scheme]
\label{rem: AmortizationSchemes}
    
The contractual notional may follow different \emph{amortization schemes}, represented by $N_{\tt{c}}(t)$, $t\in\T$, in the form of different criteria for the full repayment of the total mortgage debt, $N_{\tt{c},0}$. Common mortgage contracts with specific amortization schemes include interest-only (or bullet) mortgages, linear mortgages, and annuity mortgages \citep[see, e.g.,][]{green2014introduction}. 
The first assumes no notional repayments over the life of the mortgage, and a unique payment of $N_{\tt{c},0}$ at final time $t_n$; the second considers a linear amortization schedule over the payment dates $\T_{\tt{p}}$; the third is specified to guarantee a constant payment over the life of the mortgage comprising both interest payments and notional repayments. In either case, the amortization schemes described are piece-wise constant with discontinuities (at most) at the payment dates $\T_{\tt{p}}$.
\end{rem}
\begin{figure}[b]
    \centering
    \subfloat[\centering]{{\includegraphics[width=7.5cm]{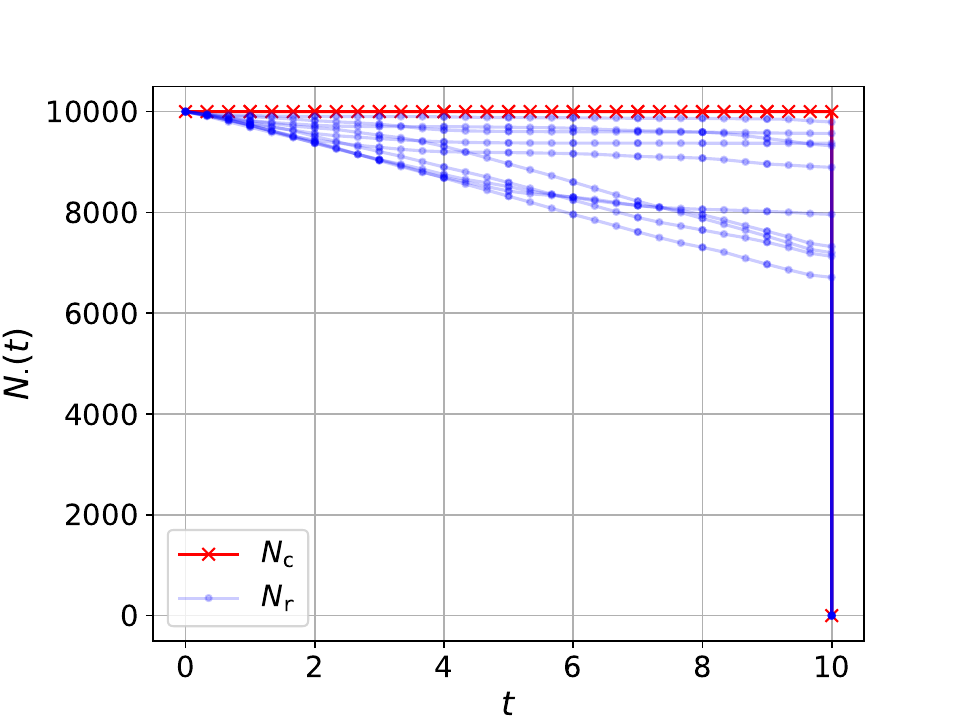} }}%
    ~\hspace{-.5cm}
    \subfloat[\centering]{{\includegraphics[width=7.5cm]{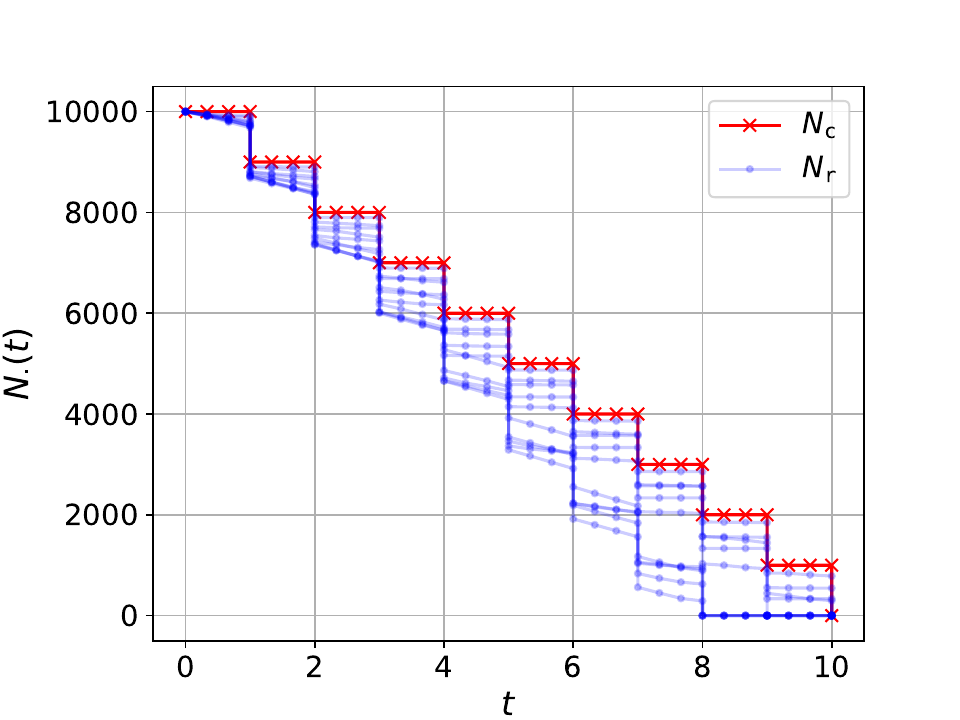} }}
    \caption{Comparison between contractual and realized notional (for different scenarios $\omega$) (a) Interest-only (bullet) mortgage. (b) Linear amortization mortgage.}
    \label{fig: NotionalProfile}%
\end{figure}

According to different regulations specific to different countries, mortgage contracts contain features that allow deviations from the contractual notional $N_{\tt{c}}(t)$. Such a mismatch between the expected amortization profile and the actual amortization schedule \emph{realized} -- on which the computation of the interest payments is based -- is due to the \emph{prepayment}. With the term ``prepayment'' we refer to all notional repayments that occur but are not scheduled in the contractual amortization scheme.
\begin{dfn}[Realized notional]
\label{def: RealisedNotional}
    Given the space of events $\Omega$ and event $\omega\in\Omega$, we indicate the \emph{realized notional under scenario $\omega$} with $N_{\tt{r}}(t,\omega)$, for $t\in\T$. 
    We use the standard notation $N_{\tt{r}}(t)$, dropping the explicit dependence on the scenario $\omega$.
\end{dfn}
The illustration of contractual and realized notional schemes, $N_{\tt{c}}$ and $N_{\tt{r}}$ for a 10-year interest-only (bullet) mortgage and a 10-year linear amortization mortgage, both with yearly payment frequency are given in \Cref{fig: NotionalProfile}. In red is presented the deterministic contractual notional, while in blue a few paths of the stochastic realized notional are shown. In \Cref{fig: NotionalProfile}b, for some paths, we visualize the \emph{early-termination} feature caused by the prepayment. Such a development is possible for any amortization scheme, but it is more likely and pronounced in contracts such as linear or annuity mortgages, where the amortization scheme includes notional repayments over the life of the contract.
\begin{dfn}[Interest payment]
\label{def: InterestPayment}
   Given a payment date $t_j\in\T_{\tt{p}}$ and a fixed interest rate $K$, the \emph{interest payment} at $t_j$ is defined as:
   \begin{equation}
   \label{eq: InterestPayment}
       I_{\cdot}(t_j;K,N_{\cdot}) = K\int_{t_{j-1}}^{t_j} N_{\cdot}(\tau)\d\tau,
   \end{equation}
   where $N_{\cdot}(t)$ can be any of the notional profiles in Definitions \ref{def: MortgageSpecs} and \ref{def: RealisedNotional} and $I_{\cdot}(t_j)$ is the corresponding interest payment.
\end{dfn}
We observe that, when the amortization schemes from \Cref{rem: AmortizationSchemes} are considered for $N_{\tt{c}}(t)$, then $I_{\tt{c}}(t_j;K,N_{\tt{c}})=KN_{\tt{c}}(t_{j-1})(t_j-t_{j-1})$ for every $t_j\in\T_{\tt{p}}$. This is true for piece-wise constant amortization schemes with discontinuities only at $t_j\in\T_{\tt{p}}$.
\begin{dfn}[Fixed-rate mortgage]
    Given the specification in \Cref{def: MortgageSpecs} and a piece-wise constant amortization scheme, a \emph{fixed rate mortgage} is a collateralized loan with the following cash flows (from the issuer perspective):
    \begin{equation}
    \label{eq: MortgageCashflows}
        \begin{aligned}
            t_0&:\quad-N_{\tt{c},0},\\
            t_j&:\quad N_{\tt{c}}(t_{j-1})-N_{\tt{c}}(t_{j}) + I_{\tt{c}}(t_j;K,N_{\tt{c}}),\qquad t_j\in\T_{p},
        \end{aligned}
    \end{equation}
    with $I_{\tt{c}}$ as in \Cref{def: InterestPayment}. In case of an interest-only mortgage, \eqref{eq: MortgageCashflows} becomes:
    \begin{equation}
    \label{eq: BulletCashflows}
        \begin{aligned}
            t_0&:\quad-N_{\tt{c},0},\\
            t_j&:\quad K N_{\tt{c},0} \Delta t_j,\qquad \quad t_j\in\T_{p}\backslash \{t_n\},\\
            t_n&: \quad N_{\tt{c},0} (1 + K \Delta t_n),
        \end{aligned}
    \end{equation}
    with $\Delta t_j=t_j-t_{j-1}$.
\end{dfn}

\begin{rem}[Fixed mortgage decomposition]
\label{rem: MortgageDecomposition}
    Every fixed-rate mortgage can be decomposed into a sum of interest-only mortgages. In particular, a fixed rate mortgage with amortization $N_{\tt{c}}$, fixed rate $K$, and payment dates $\T_{\tt{p}}$ is equivalent to the sum of $n$ interest-only mortgages with fixed rate $K$, initial notional $N_{\tt{c}}(t_{j})-N_{\tt{c}}(t_{j-1})$ and end date $t_j$, respectively.
\end{rem}
The above decomposition is convenient since it allows us to extend obvious properties of interest-only mortgages to general fixed-rate mortgages.

\subsection{Funding mechanism}

A financial institution selling a fixed-rate mortgage funds the liquidity by issuing a set of \emph{floating rate notes} (FRNs), i.e. coupon-bearing bonds whose coupon rates are given by a floating reference rate depending on the coupon payment frequency.
\begin{dfn}[FRN]
\label{def: FRN}
    An FRN written on a (constant) notional $\bar{N}$, with starting date $t_0$ and end date $t_n$ is a contract where the writer receives the notional at time $t_0$ and pays it back at time $t_n$. Over the life of the contract, the writer pays interest based on a floating reference rate. In particular, the amount of interest paid at time $t_j$ is given by:
    \begin{equation*}
        I^{fl}(t_j;L,\bar{N}) = F(t_{j-1};t_{j-1},t_j) \bar{N} \Delta t_j,\qquad  j=1,\dots,n,
    \end{equation*}
    where the forward reference rate, $F(t;t_{j-1},t_j)$, with fixing time $t_{j-1}$ for the period $(t_{j-1},t_j)$ is given by:
    \begin{equation*}
            F(t;t_{j-1},t_j)=\frac{P(t;t_{j-1})-P(t;t_{j})}{\Delta t_j P(t;t_{j})},\qquad t\leq t_{j-1},
    \end{equation*}
    and $P(t;s)$ indicates the price at time $t$ of a zero coupon bond maturing at $s\geq t$. From the issuer perspective, the cash flows of the FRN are given by:
    \begin{equation}
    \label{eq: FRNCashflows}
        \begin{aligned}
            t_0&:\quad-\bar{N},\\
            t_j&:\quad I^{fl}(t_j;F,\bar{N}),\qquad \quad t_j\in\T_{p}\backslash \{t_n\},\\
            t_n&: \quad \bar{N} + I^{fl}(t_n;F,\bar{N}).
        \end{aligned}
    \end{equation}
\end{dfn}
By coupling a fixed-rate mortgage with suitable FRNs, it is possible to exchange all notional payments for a set of floating-rate interest payments with the same amortization schedule as the mortgage. For an interest-only mortgage, it is sufficient to issue an FRN with $\bar{N}=N_{\tt{c},0}$, as illustrated in \Cref{fig: FundingMechanism} (left). The extension to general mortgages follows from \Cref{rem: MortgageDecomposition}.
When the issuer perspective is taken, the resulting portfolio corresponds to a \emph{receiver amortizing interest rate swap} (IRS) with amortizing notional given by the mortgage contractual notional, $N_{\tt{c}}$. The special case of an interest-only mortgage is illustrated in \Cref{fig: FundingMechanism}, where the resulting IRS is a vanilla instrument (no amortization). 
\begin{rem}
In a frictionless market, an FRN carries no interest-rate risk.\footnote{As stated in \citep{brigo2006interest}, in a single curve framework, ``a floating-rate note always trades at par'' at any payment dates.} Hence, the described ``funding mechanism'' suggests that a fixed rate mortgage is equivalent to a suitable (receiver amortizing) IRS in terms of the interest rate risk. This also indicates that swaps are the natural market instruments for replicating and hedging the exposure generated by fixed-rate mortgages.    
\end{rem}

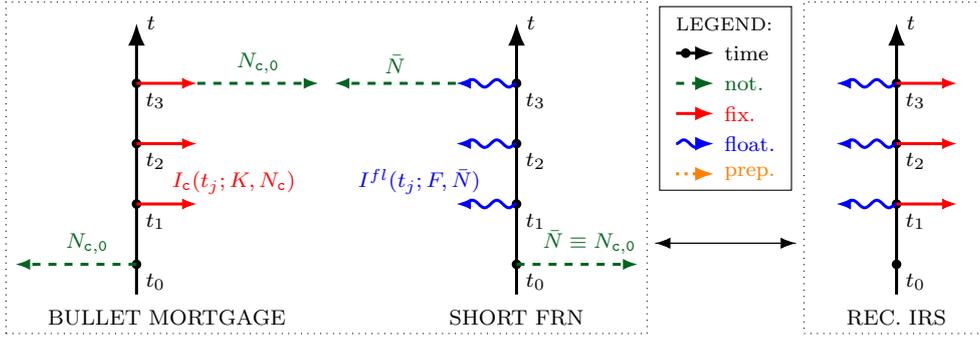
\begin{figure}[t]
    \centering
    \begin{tikzpicture}

\node[align=center] at (0.4, 0.1) {\footnotesize BULLET MORTGAGE};

\draw[-{Latex[length=3mm]}, black, line width=0.4mm] (0,0.5*0.8) -- (0,5*0.8) node[right, black] {\footnotesize $t$};

\fill[black] (0.0,1*0.8) circle (1.8pt);
\node[anchor=north west] at (0.0,1*0.8) {\footnotesize $t_0$};
\fill[black] (0.0,2*0.8) circle (1.8pt);
\node[anchor=north west] at (0.0,2*0.8) {\footnotesize $t_1$};
\fill[black] (0.0,3*0.8) circle (1.8pt);
\node[anchor=north west] at (0.0,3*0.8) {\footnotesize $t_2$};    
\fill[black] (0.0,4*0.8) circle (1.8pt);
\node[anchor=north west] at (0.0,4*0.8) {\footnotesize $t_3$};

\draw[-{Latex[length=2mm]}, darkgreen, dashed, line width=0.4mm] (0.0,1*0.8) -- (-1.6,1*0.8) node[pos=0.5, above, darkgreen] {\footnotesize $\quad N_{\tt{c},0}$};
\draw[-{Latex[length=2mm]}, darkgreen, dashed, line width=0.4mm] (0.8,4*0.8) -- (2.4,4*0.8) node[pos=0.5, above, darkgreen] {\footnotesize $N_{\tt{c},0}$};

\draw[-{Latex[length=2mm]}, red, line width=0.4mm] (0.0,2*0.8) -- (0.8,2*0.8) node[pos=1.6, above, red] {\footnotesize $I_{\tt{c}}(t_j;K,N_{\tt{c}})$};
\draw[-{Latex[length=2mm]}, red, line width=0.4mm] (0.0,3*0.8) -- (0.8,3*0.8) node[above, black] {};
\draw[-{Latex[length=2mm]}, red, line width=0.4mm] (0.0,4*0.8) -- (0.8,4*0.8) node[above, black] {};

\begin{scope}[xshift=5cm]
\node[align=center] at (0.0, 0.1) {\footnotesize SHORT FRN};
\draw[-{Latex[length=3mm]}, black, line width=0.4mm] (0,0.5*0.8) -- (0,5*0.8) node[right, black] {\footnotesize $t$};
\fill[black] (0.0,1*0.8) circle (1.8pt);
\node[anchor=north west] at (0.0,1*0.8) {\footnotesize $t_0$};
\fill[black] (0.0,2*0.8) circle (1.8pt);
\node[anchor=north west] at (0.0,2*0.8) {\footnotesize $t_1$};
\fill[black] (0.0,3*0.8) circle (1.8pt);
\node[anchor=north west] at (0.0,3*0.8) {\footnotesize $t_2$};    
\fill[black] (0.0,4*0.8) circle (1.8pt);
\node[anchor=north west] at (0.0,4*0.8) {\footnotesize $t_3$};

\draw[-{Latex[length=2mm]}, darkgreen, dashed, line width=0.4mm] (0.0,1*0.8) -- (1.6,1*0.8) node[pos=0.6, above, darkgreen] {\footnotesize $\bar{N}\equiv N_{\tt{c},0}$};
\draw[-{Latex[length=2mm]}, darkgreen, dashed, line width=0.4mm] (-0.8,4*0.8) -- (-2.4,4*0.8) node[pos=0.5, above, darkgreen] {\footnotesize $\bar{N}$};

\draw[-{Latex[length=2mm]}, blue, decorate, decoration={snake, amplitude=0.5mm, segment length=3.2mm}, line width=0.4mm] (0.0,2*0.8) -- (-0.8,2*0.8) node[pos=1.6, above, blue] {\footnotesize $I^{fl}(t_j;F,\bar{N})$};
\draw[-{Latex[length=2mm]}, blue, decorate, decoration={snake, amplitude=0.5mm, segment length=3.2mm}, line width=0.4mm] (0.0,3*0.8) -- (-0.8,3*0.8) node[above, blue] {};
\draw[-{Latex[length=2mm]}, blue, decorate, decoration={snake, amplitude=0.5mm, segment length=3.2mm}, line width=0.4mm] (0.0,4*0.8) -- (-0.8,4*0.8) node[above, black] {};
\end{scope}

\node[fit={(-1.6,0.0*0.8) (6.6,5.2*0.8)}, draw, dotted] (box1) {};

\begin{scope}[xshift=10cm]
\node[align=center] at (0.0, 0.1) {\footnotesize REC. IRS};
\draw[-{Latex[length=3mm]}, black, line width=0.4mm] (0,0.5*0.8) -- (0,5*0.8) node[right, black] {\footnotesize $t$};
\fill[black] (0.0,1*0.8) circle (1.8pt);
\node[anchor=north west] at (0.0,1*0.8) {\footnotesize $t_0$};
\fill[black] (0.0,2*0.8) circle (1.8pt);
\node[anchor=north west] at (0.0,2*0.8) {\footnotesize $t_1$};
\fill[black] (0.0,3*0.8) circle (1.8pt);
\node[anchor=north west] at (0.0,3*0.8) {\footnotesize $t_2$};    
\fill[black] (0.0,4*0.8) circle (1.8pt);
\node[anchor=north west] at (0.0,4*0.8) {\footnotesize $t_3$};

\draw[-{Latex[length=2mm]}, blue, decorate, decoration={snake, amplitude=0.5mm, segment length=3.2mm}, line width=0.4mm] (0.0,2*0.8) -- (-0.8,2*0.8) node[right, blue] {};
\draw[-{Latex[length=2mm]}, blue, decorate, decoration={snake, amplitude=0.5mm, segment length=3.2mm}, line width=0.4mm] (0.0,3*0.8) -- (-0.8,3*0.8) node[left, blue] {};
\draw[-{Latex[length=2mm]}, blue, decorate, decoration={snake, amplitude=0.5mm, segment length=3.2mm}, line width=0.4mm] (0.0,4*0.8) -- (-0.8,4*0.8) node[above, black] {\vphantom{$K N_{\tt{c},0} \Delta t$}};
\draw[-{Latex[length=2mm]}, red, line width=0.4mm] (0.0,2*0.8) -- (0.8,2*0.8) node[right, red] {};
\draw[-{Latex[length=2mm]}, red, line width=0.4mm] (0.0,3*0.8) -- (0.8,3*0.8) node[above, black] {};
\draw[-{Latex[length=2mm]}, red, line width=0.4mm] (0.0,4*0.8) -- (0.8,4*0.8) node[above, black] {};
\end{scope}

\node[fit={(8.9,0.0*0.8) (11.1,5.2*0.8)}, draw, dotted] (box2) {};

\draw[-{Latex[length=2mm]}, black] ([xshift=1mm, yshift=-10mm]box1.east) -- ([xshift=-1mm, yshift=-10mm]box2.west);
\draw[-{Latex[length=2mm]}, black] ([xshift=-1mm, yshift=-10mm]box2.west) -- ([xshift=1mm, yshift=-10mm]box1.east);

\begin{scope}[xshift=7.1cm, yshift=2.cm]
\draw[-{Latex[length=2.5mm]}, black, line width=0.4mm] (0.,1.6) -- (0.5,1.6) node[right, black] {\footnotesize time};
\fill[black] (0.125,1.6) circle (1.8pt);
\draw[-{Latex[length=2.5mm]}, darkgreen, dashed, line width=0.4mm] (0.,1.2) -- (0.5,1.2) node[right, darkgreen] {\footnotesize not.};
\draw[-{Latex[length=2.5mm]}, red, line width=0.4mm] (0.,0.8) -- (0.5,0.8) node[right, red] {\footnotesize fix.};
\draw[-{Latex[length=2.5mm]}, blue, decorate, decoration={snake, amplitude=0.5mm, segment length=3.2mm}, line width=0.4mm] (0.,0.4) -- (0.5,0.4) node[right, blue] {\footnotesize float.};
\draw[-{Latex[length=2.5mm]}, orange, dotted, line width=0.4mm] (0.0,0.0) -- (0.5,0.0) node[right, orange] {\footnotesize prep.};
\node[fit={(0.0, 1.9) (0.5, 1.9)}] {\scriptsize LEGEND:};
\end{scope}
\node[fit={(7,1.9) (8.5,4.1)}, draw] (box3) {};

\end{tikzpicture}
    \caption{\footnotesize Equivalence between an IRS (right) and the sum of a fixed rate interest-only mortgage and an FRN (left). Dashed green arrows: contractual notional payments. Solid red arrows: interest payments from a fixed rate. Curly blue arrows: interest payments from a floating rate.}
    \label{fig: FundingMechanism}
\end{figure}

Now, we take the perspective of a trader with an open position composed of ``prepayable'' fixed-rate mortgages. By monitoring her book and buying FRNs in the market, the equivalence presented above can still be achieved, however, the trader cannot purchase in advance the FRNs needed to fulfill the equivalence. Such operation has to be done for each realized market scenario, over the life of the mortgage. The following example is illustrative for the underlying idea.

\begin{exm}[Funding mechanism and ``prepayment option'']
\label{ex: FundingMechanism}
    Let us consider a simple fixed-rate mortgage with only three payment dates at $t_j=j$ years, $j=1,2,3$. We assume the fixed rate is $K$, and the notional amortization schedule is linear, starting from 1, i.e. the contractual profile is equal to:
    \begin{equation*}
        N_{\tt{c}}(t_0)=1,\quad N_{\tt{c}}(t_1)=\frac{2}{3},\quad N_{\tt{c}}(t_2)=\frac{1}{3},\quad N_{\tt{c}}(t_3)=0.
    \end{equation*}
    At the initial time, $t_0=0$, the financial institution funds the mortgage with three FRNs each one with notional $\frac{1}{3}$ and with maturity $t_j$, $j=1,2,3$, respectively (see \Cref{fig: FundingMechanismPrepayment} up-left, using the same legend as in \Cref{fig: FundingMechanism}). The financial institution's net position is a receiver IRS with linear amortizing notional $N_{\tt{c}}$, and fixed rate $K$ (see \Cref{fig: FundingMechanismPrepayment} up-right.

    Let us consider a simplified reality, where only two possible events might occur, namely $\Omega=\{\omega_0,\omega_1\}$. For the sake of the example, we assume $\omega_0$ corresponds to the event ``no prepayment occurs over the life of the contract,'' while $\omega_1$ is the event ``one prepayment of $\frac{1}{2}$ occurs at $t_1$.''
    
    Using the notation of \Cref{def: RealisedNotional}, the realized notional in the two possible scenarios is given by:
    \begin{equation*}
    \begin{aligned}
        \omega_0&:\quad N_{\tt{r}}(t_0,\omega_0)=1,\quad N_{\tt{r}}(t_1,\omega_0)=\frac{2}{3},\quad N_{\tt{r}}(t_2,\omega_0)=\frac{1}{3},\quad N_{\tt{r}}(t_3,\omega_0)=0,\\
        \omega_1&:\quad N_{\tt{r}}(t_0,\omega_1)=1,\quad N_{\tt{r}}(t_1,\omega_1)=\frac{2}{3}-\frac{1}{2}=\frac{1}{6},\quad N_{\tt{r}}(t_2,\omega_1)=0,\quad N_{\tt{r}}(t_3,\omega_1)=0.
    \end{aligned}
    \end{equation*}
    Under $\omega_1$, a trader observing the prepayment at time $t_1$ uses the unexpected notional inflow of $\frac{1}{2}$, to buy two FRNs maturing at $t_2$ and $t_3$ with notional $\frac{1}{6}$ and $\frac{1}{3}$, respectively (see \Cref{fig: FundingMechanism} down-left). By doing so, the financial institution's net position, under $\omega_1$, is an amortizing receiver IRS with \emph{different amortization scheme} from the linear version based on the contractual notional $N_{\tt{c}}$ (compare \Cref{fig: FundingMechanism} up-right and down-right).
\end{exm}
\begin{figure}[t]
    \centering
    \input{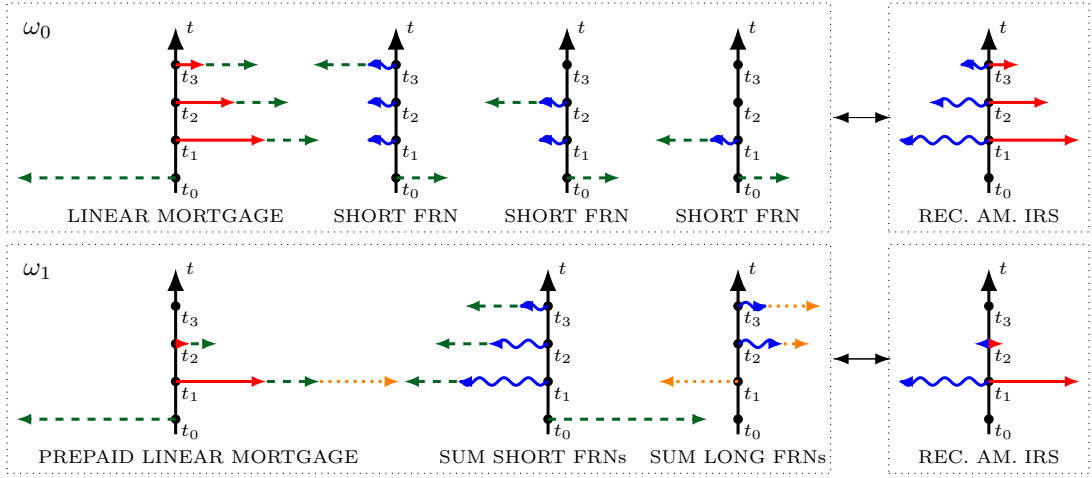}
    \caption{\footnotesize Effect of prepayment on the funding mechanism. Illustration of \Cref{ex: FundingMechanism}. Up: no prepayment, i.e. $\omega_0$. Down: prepayment, i.e. $\omega_1$. Dashed green arrows: contractual notional payments. Solid red arrows: interest payments from a fixed rate. Curly blue arrows: interest payments from a floating rate. Dotted orange arrows: notional prepayments.}
    \label{fig: FundingMechanismPrepayment}
\end{figure}
Observe that the strategy described in \Cref{ex: FundingMechanism} holds, without loss of generality, for any realized scenario bearing any prepayment behavior (this is the role of the trader who monitors her position over time), and it is independent of the choice of event space $\Omega$. In general, the equivalence between mortgages (combined with FRNs) and ``IRSs'' still holds when prepayment is considered, but a fundamental difference arises.
\begin{rem}[Stochastic notional IRS]
\label{rem: StochasticIRS}
When prepayment is contemplated, the funding equivalence only holds for each realized scenario (see \Cref{ex: FundingMechanism}), then the IRS notional is not deterministic.
Hence, the ``vanilla'' IRS used in the case of no prepayment has to be replaced with an ``exotic'' IRS, with stochastic notional. Observe that the exotic IRS' stochastic notional matches the (stochastic) realized notional $N_{\tt{r}}$ of \Cref{def: RealisedNotional}, as much as the vanilla IRS (deterministic) notional would match $N_{\tt{c}}$ of \Cref{def: MortgageSpecs}, when no prepayment is considered.
\end{rem}

An IRS with stochastic notional generalizes some well-known types of options on IRSs. For example, a European payer swaption can be interpreted as a swap with stochastic notional given by $N(t)=\Bar{N}\1_{\kappa(T)>K}$, for $\Bar{N}$ the contractual notional, $\kappa(T)$ the par swap rate\footnote{The formal definition of $\kappa$ will be provided in \eqref{def: SwapRate}.} at swaption maturity, and $K$ the contractual swap rate. A similar representation, where the stochastic notional is based on the early exercise region, holds for Bermudan swaptions.
Following this intuition, we indicate the contractual right of prepayment as an \emph{embedded\footnote{Since it is not an option explicitly written as such, but is \emph{embedded} in the mortgage contract.} prepayment option} (EPO).

\subsection{Embedded prepayment option}
To isolate the EPO's interest-rate risk, we consider the mismatch between the contractual and the realized mortgage. 
As a consequence of \Cref{rem: StochasticIRS}, the EPO is effectively an exotic IRS with stochastic notional given by the difference between the contractual and the realized mortgage notional.
\begin{dfn}[EPO]
\label{def: PrepaymentOptionPayoff}
    The EPO is an exotic IRS with stochastic notional given by:
    \begin{equation*}
        N(t)=N_{\tt{c}}(t)-N_{\tt{r}}(t), \qquad t\in\T,
    \end{equation*}
    for $N_{\tt{c}}$ and $N_{\tt{r}}$ as given in Definitions \ref{def: MortgageSpecs} and \ref{def: RealisedNotional}. The EPO cash flows read:
    \begin{equation}
    \label{eq: PrepaymentPayoff}
        \CF(t_j)= \big(K-F(t_{j-1};t_{j-1},t_j)\big)\int_{t_{j-1}}^{t_j} N(\tau) \d\tau,\qquad t_j \in\T_{\tt{p}},
    \end{equation}
    with $F(t;t_{j-1},t_j)$, as in \Cref{def: FRN}
\end{dfn}

Observe that \Cref{eq: PrepaymentPayoff} generalizes the payoff of a receiver amortizing IRS, where the notional profile is deterministic and piece-wise constant. Similarly, \Cref{eq: PrepaymentPayoff} generalizes the payoff of a receiver amortizing European swaption; in such case, the stochastic notional is piece-wise constant or zero, according to the exercise policy at maturity. 

To better understand how the prepayment notional $N(t)$ is computed, we introduce another process, $\Lambda(t)$. The stochastic process $\Lambda(t)$ represents the \emph{unconditional instantaneous prepayment rate}, given in terms of the initial contractual notional $N_0=N_{\tt{c}}(0)$. By \emph{unconditional}, we mean that $\Lambda(t)$ is the actual rate of prepayment as long as there is enough outstanding notional to prepay, otherwise, the realized rate is capped at 0, and no further prepayment occurs; the mortgage contract is terminated before the contractual end, $t_n$.

Looking at \Cref{def: RealisedNotional}, the mortgage realized notional can be written in terms of $N_{\tt{c}}$ and $\Lambda$ and reads:
\begin{equation*}
    N_{\tt{r}}(t)=\bigg[N_{\tt{c}}(t) - N_0\int_{t_0}^{t} \Lambda(\tau)\d \tau\bigg]^+,
\end{equation*}
where the operator $[\cdot]^+=\max(0, \cdot)$ is introduced to avoid a negative realized notional. Similarly, the prepayment notional reads:
\begin{equation}
\label{eq: PrepaymentNotional}
\begin{aligned}
    N(t) &= N_{\tt{c}}(t)-N_{\tt{r}}(t)\\
     &= \min\bigg(N_{\tt{c}}(t),N_0\int_{t_0}^{t} \Lambda(\tau)\d \tau\bigg).
\end{aligned}
\end{equation}
From \Cref{eq: PrepaymentNotional}, the source of randomness in the notional $N(t)$ depends on the process $\Lambda(\tau)$, $\tau\leq t$. The next section provides details on its specifications.

\section{Prepayment model and replication}
\label{sec: PrepaymentModel}

\subsection{Stochastic prepayment functional form}
\label{ssec: PerturbedRIFunction}

A standard approach in the literature specifies a suitable deterministic mapping between a set of explanatory variables -- possibly modeled as stochastic processes -- and the rate of prepayment $\Lambda$ \citep[see, e.g.,][]{hayre2003prepayment,casamassima2022pricing,jagannathan2022machine}. 
The main explanatory variable is the \emph{financial rate incentive} (RI), indicated by a process $\varepsilon(t)$, representing how appealing for prepayment the current market level is.
The RI is defined as a function of a (stochastic) market risk factor, say $r(t)$. Notably, a common choice is to model $\varepsilon(t)$ as the spread between the fixed interest rate $K$ of the original mortgage contract and the prevailing (fixed) rate $\kappa(t)$ offered in the market for a new mortgage with the same features and outstanding tenor as the original contract \citep[see][]{perry2001study,hoda2007implementation}, i.e.:
\begin{equation}
\label{eq: RateIncentive}
    \varepsilon(t) = K - \kappa(t) =: g(r(t)),
\end{equation}
for some suitable, deterministic function $g$. We adopt this assumption here, but another used choice for the rate incentive is the ratio between $K$ and $\kappa(t)$ \citep[see, e.g.,][]{richard1989prepayments}. Furthermore, by non-arbitrage arguments, it is easy to show that the fixed rate currently offered in the market for a new mortgage is equal -- in a mathematically fair setting -- to the par rate of a swap with the same notional scheme and frequency as the outstanding notional. 
\begin{dfn}[Generalised swap rate]
\label{def: SwapRate}
    Let $N_{\tt{c}}(t)$, $t\in\T=[t_0,t_n]$, be the deterministic piece-wise constant contractual notional scheme for an IRS with reset dates $\T_{\tt{r}}=\{t_0,\dots,t_{n-1}\}$ and payment dates $\T_{\tt{p}}=\{t_1,\dots,t_{n}\}$. Then, we define the stochastic \emph{swap rate} $\kappa$ as:
    \begin{equation*}
        \kappa(t)=\frac{1}{A(t)}\sum_{\substack{t_j\in\T_{\tt{p}}\\ t_j> t}}N_{\tt{c}}(t_{j-1})\big(P(t;t_{j-1})-P(t;t_j)\big),\qquad t\in\T,
    \end{equation*}
    with $P(t;t_j)$ as in \Cref{def: FRN} and \emph{annuity factor}:
    \begin{equation*}
        A(t)=\sum_{\substack{t_j\in\T_{\tt{p}}\\ t_j> t}}N_{\tt{c}}(t_{j-1})\Delta t_j P(t;t_{j}).
    \end{equation*}
\end{dfn}
Note that when the swap notional, $N_{\tt{c}}$, in \Cref{def: SwapRate} is constant, we end up with the simplified expression for a vanilla swap rate, where most of the terms in the numerator disappear because of a telescopic sum \citep[see, e.g.,][]{brigo2006interest,oosterlee2019mathematical}.

RI $\varepsilon(t)$ is mapped into $\Lambda(t)$ through the following relationship, in our setting,
\begin{equation}
\label{eq: CPRSigmoid}
\begin{aligned}
        \Lambda(t)&={h}_{{RI}}(t, \varepsilon(t);l,u,a,b),\\
    {h_{RI}}(t, x;l,u,a,b)&=l+\frac{u-l}{2}\Big[\tanh{\big(a(x+b)\big)}+1\Big],
\end{aligned}
\end{equation}
with $\varepsilon(t)$ as in \eqref{def: SwapRate}, and the parameters $l$, $u$, $a$ and $b$. The first two parameters control the lower and upper bounds for the prepayment amount, while $a$ and $b$ control the client's reaction rate to the RI.
A monotonically increasing function $h_{RI}$ has been chosen since the mortgage owners are more prone to prepayment the higher the RI is. Furthermore, empirical evidence shows an unelastic (or delayed) reaction to a positive RI, meaning that a continuous function has to be preferred to a step function \citep{casamassima2022pricing} (see \Cref{fig: Sigmoid}a).

\begin{figure}[t!]%
    \centering
    \subfloat[\centering]{{\includegraphics[width=7.15cm]{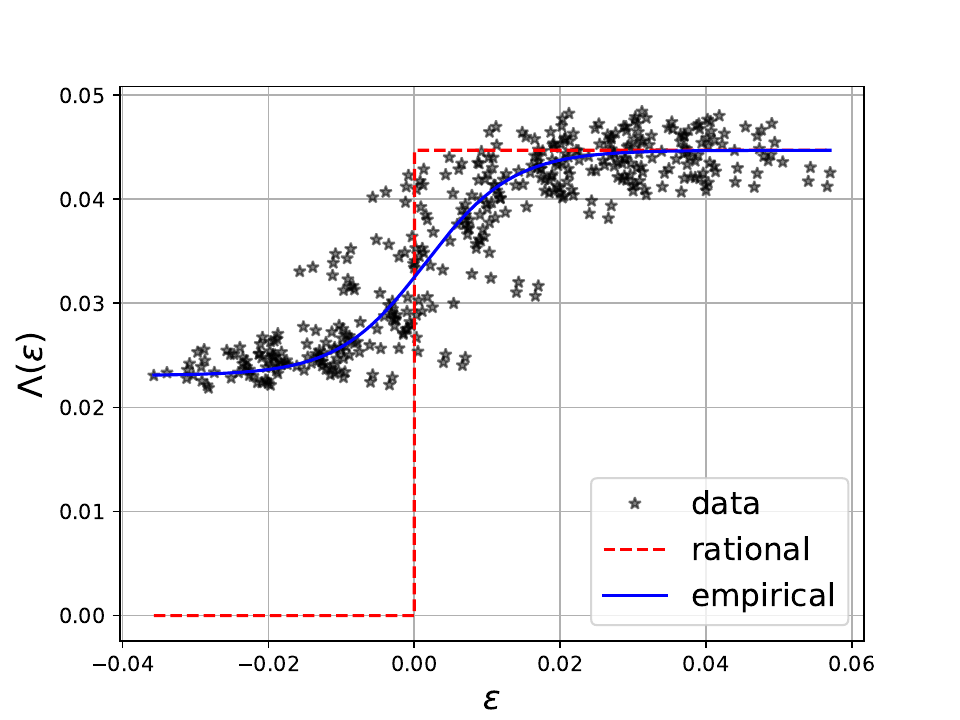} }}%
    ~
    \subfloat[\centering]{{\includegraphics[width=6.5cm]{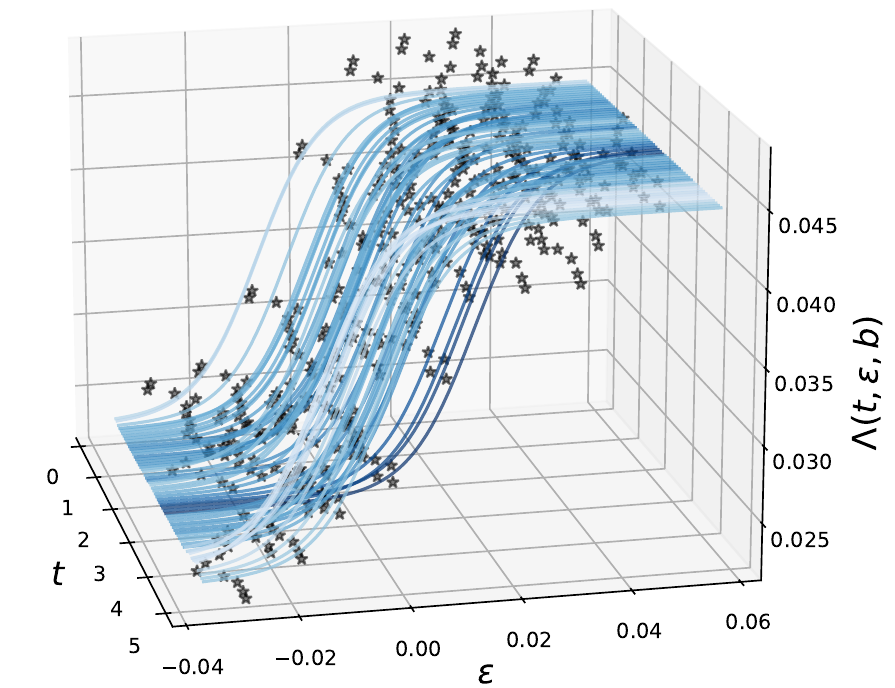} }}\\
    \subfloat[\centering]{{\includegraphics[width=7.cm]{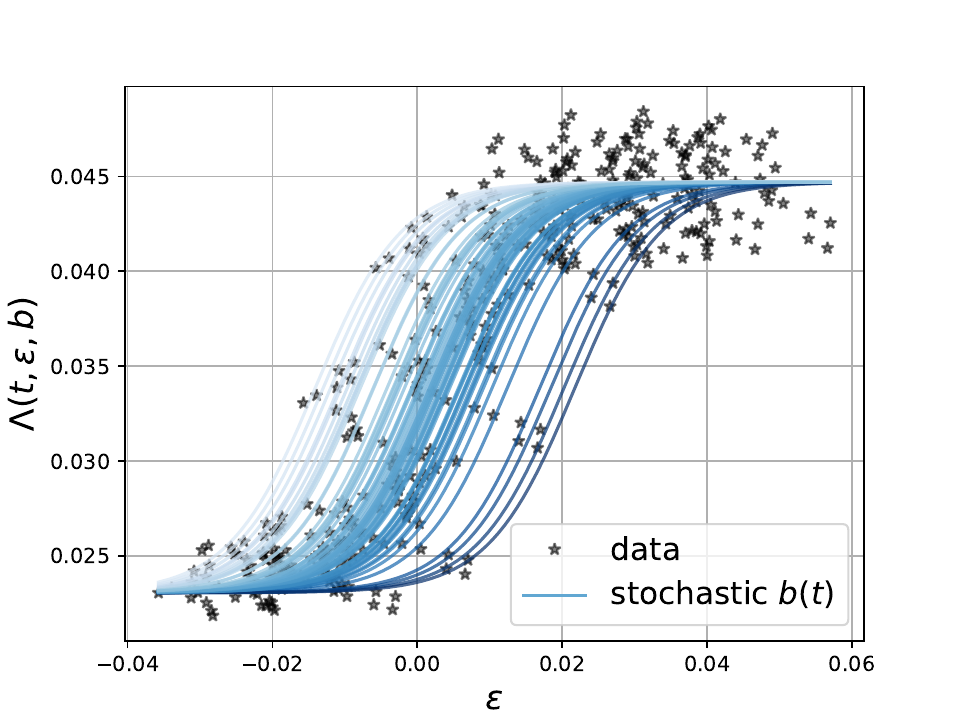} }}%
    ~
    \subfloat[\centering]{{\includegraphics[width=7.1cm]{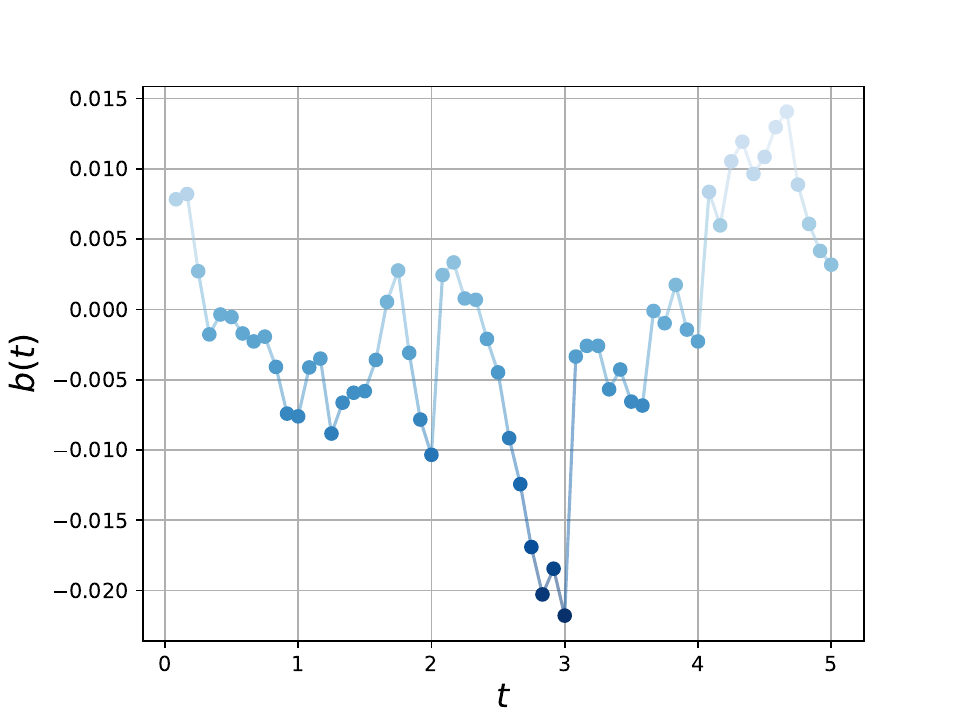} }}%
    \caption{(a) Rational (dashed red line) and irrational (solid blue line) rate incentive functions. Black stars represent the prepayment data. (b) Stochastic sigmoid function over time (solid lines in different gradations of blue), with prepayment data (black stars). (c) Projection of (b). (d) Time series for risk factor $b(t)$ corresponding to the sigmoid horizontal shift in (b) and (c).}%
    \label{fig: Sigmoid}%
\end{figure}

A deterministic relationship between $\varepsilon(t)$ and $\Lambda(t)$ is based on the assumption that the prepayment behavior only depends on the state of the market. We call such an assumption the \emph{deterministic rate-to-prepayment mapping assumption} (DMA). 
The main drawback of the DMA is neglecting that part of the prepayment behavior is not driven by ``financially rational'' variables.

Relaxing the DMA is one of the main contributions of this research to the existing literature. Particularly, we consider a stochastic $b(t)$ in \eqref{eq: CPRSigmoid}, obtaining a second risk factor $b(t)$ to account for the uncertainty embedded in the prepayment behavior. 
Risk factor $b(t)$ can be interpreted as a spread on the financial rate incentive $\varepsilon(t)$ in \eqref{eq: RateIncentive} aiming to capture deviations in the prepayment behavior compared to the financially rational choice. The combined quantity $\varepsilon(t)+b(t)$ is the \emph{perceived rate incentive} and in general, it can differ from the financial rate incentive.

Combining \Cref{eq: RateIncentive,eq: CPRSigmoid}, we write $\Lambda(t)$ as a function $h$ of the two risk factors $r(t)$ and $b(t)$, i.e:
\begin{equation}
\label{eq: CPRSigmoid2RiskFactors}
\begin{aligned}
        \Lambda(t)&={h}(t, r(t),b(t)),\\
        h(t, x,y)&={h}_{RI}(t, g(x);l,u,a,y),
\end{aligned}
\end{equation}
where we drop the explicit dependence on the structural parameters $l$, $u$, and $a$ to keep a simpler notation.

We remark that \Cref{eq: CPRSigmoid2RiskFactors} generalizes (\ref{eq: CPRSigmoid}) which is recovered as a limit case when $b(t)$ does not depend on the scenario $\omega$.
The stochastic coefficient $b$ gives rise to a time series of rate incentive functions for each realized scenario (see \Cref{fig: Sigmoid}b) resulting in a horizontal shift from a baseline sigmoid, see \Cref{fig: Sigmoid}c. This makes the model less sensitive to historical calibration, resulting in a robust framework.
Introducing a second stochastic risk factor to encapsulate uncertainty that market risk factors cannot explain constitutes a novelty in the framework of fixed-income option evaluation when a behavioral component is assessed. 

\subsection{Risk-neutral evaluation of the embedded prepayment option}
\label{ssec: RiskNeutralEvaluationEPO}

Consistently with the standard pricing literature, the EPO value process is evaluated under a risk-neutral measure $\Q$ and requires knowledge of the risk-neutral dynamics for the two risk factors $r(t)$ and $b(t)$.\footnote{More precisely, risk-neutral dynamics are needed when we intend to perform pricing using the money savings account as num\'eraire. In general, for any choice of num\'eraire, the pricing dynamics must be such that any tradable asset price process relative to the chosen num\'eraire is a martingale.} The model choice for the market risk factor $r(t)$ is the well-known Hull-White short rate model \citep{hull1990pricing},\footnote{The well-known model -- among both academics and practitioners -- that belongs to the arbitrage-free class of dynamics generalized by \citep*{heath1992bond}.} while the dynamics for $b(t)$ are modeled as an Ornstein–Uhlenbeck (OU) process \citep{uhlenbeck1930theory}, leading to the system of SDEs:
\begin{align}
\label{eq: RiskNeutralDynamicsRate}
    \d r(t) &= \alpha_r^{\Q} (\vartheta_r^{\Q}(t) - r(t)) \d t + \eta_r \d W^{\Q}_r(t),\qquad r(t_0)=r_0\in\R,\\
\label{eq: RiskNeutralDynamicsPrepayment}
    \d b(t) &= \alpha_b^{\Q} (\theta_b^{\Q} - b(t)) \d t + \eta_b \d W^{\Q}_b(t),\qquad \quad \;b(t_0)=b_0\in\R,
\end{align}
for $\alpha_r^{\Q},\alpha_b^\Q,\eta_r,\eta_b>0$, $\theta_b^{\Q}\in\R$, and $\vartheta_r^\Q(t)$ time-dependent. $W_r^\Q$ and $W_b^{\Q}$ are Wiener processes such that $\d W_r^\Q \d W_b^\Q = \rho\d t$ for $\rho\in[-1,1]$.
A mean-reverting process for $b(t)$ is chosen to include a controlled level of diffusion representing uncertainty in the prepayment behavior. Observing the time series for $b(t)$ from the realistic data (see \Cref{fig: Sigmoid}d), the behavior of $b$ is not purely ``diffusive'' but tends to fluctuate around a mean value. Furthermore, we remark that, even though $\rho$ can attain negative values, we expect it to be positive. A positive correlation with $r$ entails a negative correlation with $\varepsilon$, which is realistic since $b$ acts as an antagonist against the financially rational rate incentive, $\varepsilon$.

Given \eqref{eq: RiskNeutralDynamicsRate} and \eqref{eq: RiskNeutralDynamicsPrepayment}, and using $M(t)=\exp\int_{t_0}^t r(\tau)\d\tau$, the money savings account, the value process for the EPO reads, for $t\in\T$:
\begin{equation}
    \label{eq: PrepaymentValue}
    V(t)=\E_t^{\Q}\Bigg[\sum_{t_j \geq t}\frac{M(t)}{M(t_j)}\CF(t_j)\Bigg],
\end{equation}
where the cash flows $\CF(t_j)$ are defined in \eqref{eq: PrepaymentPayoff}, $(\F(t))_{t\in\T}$ is the augmented natural filtration generated by the two Wiener processes, $W_r^\Q$ and $W_b^{\Q}$, and $\E_t^\Q$ is the short notation for the conditional expectation with respect to $\F(t)$ under the measure $\Q$.

\begin{rem}[Real-world and risk-neutral dynamics]
\label{rem: ConnectionRealWorldRiskNeutralGirsanov}
    The parameters in \eqref{eq: RiskNeutralDynamicsRate} are calibrated based on market information and are uniquely determined. However, since there is no market information for the non-observable risk factor $b(t)$, the dynamics for $b(t)$ in \eqref{eq: RiskNeutralDynamicsPrepayment} are based on the history of prepayment. In formal terms, the information available is given by:
    \begin{equation}
        \label{eq: RealWorldDynamicsPrepayment}
    \d b(t) = \alpha_b^{\P} (\theta_b^{\P} - b(t)) \d t + \eta_b \d W^{\P}_b(t),
    \end{equation}
    where the parameters are obtained via historical calibration.
    
    In general, \eqref{eq: RealWorldDynamicsPrepayment} is not under risk-neutral dynamics and shall not be used for pricing. The Girsanov's theorem \citep{girsanov1960transforming} provides the mathematical toolbox to recover the $\Q$-dynamics in \eqref{eq: RiskNeutralDynamicsPrepayment} from the $\P$-dynamics in \eqref{eq: RealWorldDynamicsPrepayment}. Indeed, as long as $\P$ and $\Q$ are equivalent measures, it is possible to connect the two Wiener processes, $W_b^\P(t)$ and $W_b^\Q(t)$, using a suitable process $\lambda(t)$. In particular, the well-known relation reads:
    \begin{equation}
    \label{eq: MarketPriceOfRisk}
        \d W_b^\Q(t)=\d W_b^\P(t) + \lambda(t)\d t,
    \end{equation}
    and $\lambda(t)$ is the \emph{market price of risk} \citep[see, e.g.,][]{lintner1970market}.
\end{rem}

From \Cref{rem: ConnectionRealWorldRiskNeutralGirsanov}, we deduce that, because of the non-observable factor $b(t)$, the value process for the EPO is \emph{only} determined once the market price of risk is specified, and -- as a consequence -- is non-unique. To underline this aspect, we will explicitly report the dependence of the value process on the market price of risk $\lambda(t)$, namely $V(t)\equiv V(t;\lambda)$. The price dependence on the market price of risk $\lambda$ can be interpreted as the dependence of prices on the \emph{risk aversion parameter} in the theory of \emph{utility indifference pricing}. We refer to \citep[Chapter 2]{carmona2009indifference} for a detailed exposition on the topic.

\subsection{Effect of uncertainty on the embedded prepayment option}
\label{ssec: LowPriceEPO}

The inclusion of behavioral uncertainty reduces the price of the EPO, as shown in the following theorem for the special case of a bullet mortgage and stepwise incentive function (i.e. $a\to +\infty$) with $l=0$ and $u\leq \frac{1}{n}$. The incentive function described is the \emph{rational} version and the constraint on $u$ ensures the notional availability to be prepaid over the contract's life. Furthermore, we assume the prepayment only occurs at the reset dates $\T_{\tt{r}}=\{t_0,\dots,t_{n-1}\}$ based on the current realization of the economy $(r(t_j),b(t_j))$.

\begin{thm}
\label{thm: DecreasingPrepaymentPrice}
    We consider the map:
    \begin{equation*}
    v:\: \R^+\longrightarrow \R, \qquad \eta_b\longmapsto V(t_0;\eta_b),   
    \end{equation*}
    for $V(t_0;\eta_b)\equiv V(t_0)$ defined in \Cref{eq: PrepaymentValue} and $\eta_b$ the volatility coefficient in \eqref{eq: RiskNeutralDynamicsPrepayment}. Let us assume that:
    \begin{enumerate}
        \item \label{enum: InterestOnly} The mortgage amortization scheme is interest-only with fixed interest rate $K>0$;
        \item \label{enum: RationalIncentive} $a\to +\infty$ and $l=0$  in \eqref{eq: CPRSigmoid}, i.e. a rational rate incentive function;
        \item \label{enum: NotionalAvailability} $u\leq \frac{1}{n}$ in \eqref{eq: CPRSigmoid}, i.e. notional availability;
        \item \label{enum: PaymentsAtReset} Prepayments only occur at reset dates $\T_{\tt{r}}$;
        \item \label{enum: IndependentB} $r$ and $b$ in \eqref{eq: RiskNeutralDynamicsRate} and \eqref{eq: RiskNeutralDynamicsPrepayment} are uncorrelated with $b_0=\theta_b^\Q=0$.
    \end{enumerate}
     Then, $V(t_0;\eta_b)$ is decreasing in the volatility $\eta_b$.
\end{thm}
\begin{proof}
    The proof is given in Appendix \ref{app: ProofTheorem}.
\end{proof}

\begin{cor}
\label{cor: CheapPrepaymentPrice}
    Under the same hypothesis as in \Cref{thm: DecreasingPrepaymentPrice}, the cost of prepayment obtained after relaxing the DMA, i.e. when a non-trivial stochastic risk factor $b$ is included, is lower than the cost computed with the DMA of the rate incentive into the fraction of prepayment.
\end{cor}
\begin{proof}
    Observing that the DMA is equivalent to a zero volatility $\eta_b$, the proof is a trivial application of \Cref{thm: DecreasingPrepaymentPrice}.
\end{proof}

\begin{rem}
    The result in \Cref{thm: DecreasingPrepaymentPrice} holds in the more general case where the prepayment dates and the amount prepaid at each date are assumed to be known a priori. The intuition is that as long as we manage to represent the EPO value as a sum of European-type swaptions with exercise region $\kappa(t_k)+b(t_k)<K$, for $k=0,\dots,n-1$,\footnote{As done in the proof of \Cref{thm: DecreasingPrepaymentPrice}.} any uncertainty in $b(t_k)$ entails a suboptimal exercise, and, as a consequence, a lower value. When the notional availability is path-dependent and \emph{not known} in advance (hence, we might ``run out'' of notional to prepay), the $b(t_k)$ term in the exercise region may act as a correction against a suboptimal ``early-exercise'' of the EPO (that would coincide with the exercise of a standard European-type swaption).
\end{rem}

\subsection{Replication of the embedded prepayment option}
\label{ssec: HedgingPrepayment}

The stochastic risk factor $b$ is \emph{non-observable} and inherently \emph{non-hedgable}, meaning that we are in the framework of an \emph{incomplete economy} \citep{bjork2019arbitrage}. 
The presence of the non-hedgeable risk factor $b$ (combined with discrete monitoring) makes the perfect replication of the EPO infeasible. Indeed, no tradable instrument can mirror the uncertainty in the EPO value since no market instrument depends on $b$ itself. Taking a risk management perspective, the hedging problem is formalized as an optimization problem whose solution is the optimal, static replicating strategy that minimizes the path-wise exposure of the EPO -- for a certain monitoring time window -- by investing in tradable instruments, such as IRSs and swaptions. 

\subsubsection{Objective of the replication}

Let us introduce a set of $I\in\NN\backslash\{0\}$ market instruments available for the replication.\footnote{In the specific case we will consider interest rate swaps and swaptions.} For every instrument, $i=1,\dots,I$, we indicate its value process as $S_i(t)$, $t\in\T$. Furthermore, the static number of units the hedger holds in each instrument is $w_i$. 

The objective is to find the optimal allocation, $\w=[w_1,\dots,w_I]^\top$, to minimize some metric of interest based on the path-wise distance between the wealth invested in the EPO and the wealth invested in the replicating instruments. Since the considered instruments bear cash flows over their life, we have to account for those as well. We assume every cash flow is immediately invested in a risk-free money-savings account with yearly continuously compounded return given by $r(t)$.
For a specific market price of risk $\lambda(t)$ in \eqref{eq: MarketPriceOfRisk}, we define the two wealth processes for the EPO and the replicating strategy, respectively, as:
\begin{align}
\label{eq: WealthEPO}
    \W_{V}(t; \lambda) &= V(t;\lambda) + C_{V}(t; \lambda),\\
    \label{eq: WealthHedge}
    \W_S(t;\w) &= \sum_i w_i \big(S_i(t) + C_i(t)\big)= \sum_i w_i \W_i(t),
\end{align}
where $C_V(t; \lambda)$, $C_i(t)$ are cash accounts where the cash flows are collected and invested at the risk-free rate, $\W_i(t)$ is the wealth at time $t$ generated by purchasing one unit of the $i$-th instrument at time $t_0$, and $V(t; \lambda)=V(t)$ in \eqref{eq: PrepaymentValue}

We introduce a \emph{signed distance} process $\D$ as:
\begin{equation}
    \label{eq: SignedDistance}
    \D(t;\w,\lambda) = \W_V(t;\lambda) - \W_S(t;\w),
\end{equation}
bearing the information on the path-wise mismatch between the exposure (EPO) and the replication.

Based on the process in \eqref{eq: SignedDistance}, we may define different loss functions targeting various aspects of the exposure that may be relevant. 
For instance, using $M^p(X)$, $p\in\NN$, to indicate the standard $p$-order absolute moment of a (sufficiently integrable) random variable $X$, i.e. $M^p(X)=\E^\cdot[|X|^p]$, we can enforce path-wise matching by considering a loss function given by:
\begin{equation}
\label{eq: LossObjectiveMoments}
    \L_{M^p}(\w,\lambda)= \int_{t_0}^{T} \alpha(t)M^p\big(\D(t;\w,\lambda)\big)\d t,
\end{equation}
where $\alpha(t)$, with $\int_{t_0}^T \alpha(t)\d t=T - t_0$, is a deterministic weight controlling the focus of the optimization along the time axis, while the choice of $p$ affects the importance given to the tails of the distance distribution compared to the center.

Another specification may enforce focus on the right tail of the signed distance distribution by including in the loss function the right expected shortfall to a certain level $q\in (0,1)$, namely:
\begin{equation}
\label{eq: LossObjectiveExpectedShortfall}
    \L_{ES^+_q}(\w,\lambda)= \int_{t_0}^{T} \alpha(t)ES^+_q\big(\D(t;\w,\lambda)\big)\d t,
\end{equation}
where $ES^+_q(X)=\E^\cdot[X|X>Q_X(q)]$, with $Q_X$ the quantile function of $X$. Similarly, we may define a loss function based on the left expected shortfall $ES^-_q(X)=\E^\cdot[X|X<Q_X(q)]$.

In general, based on the purpose of the replication task, we can combine different metrics to prioritize different regions of the distance distribution in the optimization process. By combining the losses defined in \eqref{eq: LossObjectiveMoments} and \eqref{eq: LossObjectiveExpectedShortfall}, we obtain the objective:
\begin{equation}
\label{eq: LossObjectiveGeneral}
    \L_{p,q}(\w,\lambda)= \L_{M^p}(\w,\lambda) + k\cdot\L_{ES^+_{q}}(\w,\lambda).
\end{equation}
Here, the constant $k>0$ controls the two effects and has to be tuned.
For instance, fixing $p=2$ and $q=0.9$, the loss function obtained penalizes both the mean squared error and the right tail error more extreme than the 90\% quantile, ensuring a good path-wise replication and avoiding large deviations between the EPO exposure and the replication. 

\subsubsection{Conditional and unconditional replication}

Through the EPO wealth process $\W_V(t;\lambda)$, every loss function depends on the specification of the market price of risk $\lambda(t)$ (see \Cref{rem: ConnectionRealWorldRiskNeutralGirsanov}). This implies that the direct minimization of the loss function would result in a \emph{conditional optimal strategy}. Here, ``conditional'' means that the result of the optimization is optimal only if the specified market price of risk is realized. Formally, we set the minimization problem:
\begin{equation}
\label{eq: MinimizationProblem}
    \underset{\w\in\R^I}{\min} \:\L_{p,q}(\w,\lambda),
\end{equation}
for the loss function $\L_{p,q}$ defined in \Cref{eq: LossObjectiveGeneral}. If a solution exists, we indicate it with $\w^*(\lambda)=\arg \min_{\w}\L_{p,q}(\w,\lambda)$ to underline the dependence on a specific market price of risk. In other words, under $\lambda$, the optimal strategy $\w^*(\lambda)$ is the vector $\w$ that attains the loss $\L_{p,q}(\w,\lambda)$ minimum.

Nonetheless, because of the non-observable nature of $b(t)$, $\lambda(t)$ is in general not known. Hence, we also define a different optimization problem leading to a \emph{robust} strategy against a possible misspecification of the market price of risk. With this purpose in mind, a robust replication problem is recast as:
\begin{equation}
\label{eq: MinMaxProblem}
    \underset{\w\in\R^I}{\min} \:\underset{\lambda\in\M}{\max}\:\L_{p,q}(\w,\lambda),
\end{equation}
where $\M$ is the domain of the possible market prices of risk. 
A solution $(\lambda^*,\w^*)$ of \eqref{eq: MinMaxProblem} is \emph{robust} in the following sense. Consider a specification of the market price of risk, $\lambda^*$, and of the optimal robust strategy, $\w^*$. Such a pair has a loss equal to $L^*=\L(\w^*,\lambda^*)$. Then, $\w^*$ is (locally)\footnote{Without specific assumption on the objective function, the property only holds locally. The convexity/concavity of the objective function determines such a region.} \emph{robust} because any (small) misspecification of the market price of risk (i.e., (small) variation around $\lambda^*$) leads to a loss smaller than $L^*$, namely to a better performance of the hedging strategy.
Effectively, solving \eqref{eq: MinMaxProblem} corresponds to finding the best replication that minimizes the loss in a locally-worst-case scenario as, for instance, in \citep{balter2020pricing} where the problem of hedging a liability defined in an incomplete economy is addressed.

Neither the pricing in \eqref{eq: PrepaymentValue}, nor the optimizations in \eqref{eq: MinimizationProblem} and \eqref{eq: MinMaxProblem} allow for an analytic solution, the task is solved numerically, which is the purpose of the next section.

\section{Numerical pricing and replication}
\label{sec: NumericalPricingReplication}

\subsection{EPO pricing}
\label{ssec: NumericalPricing}

The EPO value process defined in \Cref{eq: PrepaymentValue} cannot be computed in analytic form due to the nonlinear, path-dependent payoff given in \eqref{eq: PrepaymentPayoff}. We will use a modification of the \emph{least squares Monte Carlo} (LSM) method introduced by \citep{longstaff2001valuing} in the framework of pricing American/Bermuda options. 
The modified LSM pricing scheme is obtained based on the following proposition.

\begin{prop}
\label{prop: RecursiveRelationEPO}
    Let $\tn(t)=\inf\{t_j\geq t: t_j\in\T_{\tt{p}}\}$ with $\T_{\tt{p}}$ the set of payment dates defined in \eqref{eq: PaymentDates}, and let $\tp(t)=\sup\{t_j< \tn(t): t_j\in\T_{\tt{r}}\}\}$ with $\T_{\tt{r}}=\{t_0,\dots,t_{n-1}\}$ the set of reset dates of the reference forward rate $F(t;t_{j-1},t_j)$ in \Cref{def: FRN}. By convention, if $\{t_j\geq t: t_j\in\T_{\tt{p}}\}=\emptyset$, then $\tn(t)=\tp(t)=t$. \\
    For $t\in\T=[t_0,T]$, let us define the auxiliary quantities $\CFk(t)$, $\CFu(t)$, and $\FV(t)$ as:
    \begin{align}
    \label{eq: IntegralKnown}
        \CFk(t)&=P(t;\tn(t))\int_{\tp(t)}^tN (\tau)\d \tau,\\
            \label{eq: IntegralUnknown}
        \CFu(t)&=\E_t^\Q\bigg[\frac{M(t)}{M(\tn(t))}\int_{t}^{\tn(t)} N (\tau)\d \tau\bigg],\\
            \label{eq: FutureCashflowsValue}
        \FV(t)&=\E_t^\Q\bigg[\sum_{t_j > \tn(t)}\frac{M(t)}{M(t_j)}\CF(t_j)\bigg].
    \end{align}
    Then, for $t\in\T$, the EPO value $V(t)$ is given by:
    \begin{equation}
    \label{eq: PricingFormula}
        V(t)=\big(K-F(\tp(t);\tp(t),\tn(t))\big)\big(\CFk(t) + \CFu(t)\big)+\FV(t).
    \end{equation}
    Furthermore, given $t_+>t$, the following recursion holds:
    \begin{enumerate}
        \item When $\tn(t)=\tn(t_+)$, namely $t<t_+\leq \tn(t)$:
    \begin{align}
    \label{eq: RecursionH}
        \CFu(t) &= \E_t^\Q\bigg[\frac{M(t)}{M(t_+)}\Big(P(t_+;\tn(t_+))\int_{t}^{t_+} N(\tau)\d\tau + \CFu(t_+)\Big)\bigg],\\
        \label{eq: RecursionF}
        \FV(t) &= \E_t^\Q\bigg[\frac{M(t)}{M(t_+)}\FV(t_+)\bigg].
    \end{align}
    \item When $\tn(t)=\tp(t_+)$ and $t=\tn(t)<t_+$:
    \begin{align}
    \label{eq: RecursionHTrivial}
        \CFu(t) &= 0,\\
        \label{eq: RecursionV}
        \FV(t) &= \E_t^\Q\bigg[\frac{M(t)}{M(t_+)}V(t_+)\bigg].
    \end{align}
\end{enumerate}
    
\end{prop}

\begin{proof}
    The proof is given in Appendix \ref{app: ProofProposition}.
\end{proof}

Based on \eqref{eq: PricingFormula} in \Cref{prop: RecursiveRelationEPO}, we define a pricing routine to compute the value paths of the EPO. The task requires computing the conditional expectation $\E_t^\Q$ appearing in the definitions of $\CFu$ and $\FV$, \eqref{eq: IntegralUnknown} and \eqref{eq: FutureCashflowsValue}, respectively. When the quantities of interest are not path-dependent and the underlying state variables are Markov processes, then the expectation conditional to $\F(t)$ is approximated via regression onto a subspace spanned by the state variables at time $t$. Even though \eqref{eq: IntegralUnknown} and \eqref{eq: FutureCashflowsValue} involve path-dependent quantities with respect to $r$ and $b$, we can overcome such a problem by including the notional $N$ in the set of the state variables. Since $N$ is a Markov process, given the vector of state variables, $X(t)=[r(t), b(t), N(t)]^\top$, \eqref{eq: IntegralUnknown} and \eqref{eq: FutureCashflowsValue} read:
\begin{align}
\label{eq: DoobH}
    \CFu(t)=\E^\Q\bigg[\frac{M(t)}{M(\tn(t))}\int_{t}^{\tn(t)} N (\tau)\d \tau\Big|X(t)\bigg]&=:h_t(X(t))=\sum_\ell \beta^{h,t}_{\ell} \psi_\ell(X(t)),\\
    \label{eq: DoobF}
    \FV(t)=\E^\Q\bigg[\sum_{t_j > \tn(t)}\frac{M(t)}{M(t_j)}\CF(t_j)\Big|X(t)\bigg]&=:f_t(X(t))=\sum_\ell \beta^{f,t}_{\ell} \psi_\ell(X(t)),
\end{align}
for appropriately $\sigma(X(t))$-measurable functions $h_t$ and $f_t$ \citep[Doob's measurability criterion]{baldi2017stochastic}, and a suitable basis $(\psi_\ell)_\ell$ for the space of the $\sigma(X(t))$-measurable functions.
In general, $h_t$ and $f_t$ are not known but can be approximated by truncation of the right-hand side of \eqref{eq: DoobH} and \eqref{eq: DoobF}:
\begin{align}
\label{eq: TruncatedDoobH}
    \CFu(t)=h_t(X(t))\approx\sum_{\ell=0}^L \beta^{h,t}_{\ell} \psi_\ell(X(t))&=:\widehat{h_t}(X(t)),\\
    \label{eq: TruncatedDoobF}
    \FV(t)=f_t(X(t))\approx \sum_{\ell=0}^L \beta^{f,t}_{\ell} \psi_\ell(X(t))&=:\widehat{f_t}(X(t)).
\end{align}

The coefficients $(\beta^{h,t}_{\ell})_{\ell=0}^L$ and $(\beta^{f,t}_{\ell})_{\ell=0}^L$ are the results of least squares regression using the recursion in \Cref{eq: RecursionH,eq: RecursionF,eq: RecursionHTrivial,eq: RecursionV}. This method, introduced in \citep{longstaff2001valuing}, is based on the geometric interpretation of the conditional expectation. Particularly, for generic random variables $X$ and $Y$, $\E[Y|X]$ is the orthogonal $L^2$-projection of $Y$ onto the subspace of $\sigma(X)$-measurable random variables \citep[Remark 4.3]{baldi2017stochastic}. As a consequence, $(\beta^{h,t}_{\ell})_{\ell=0}^L$ and $(\beta^{f,t}_{\ell})_{\ell=0}^L$ are obtained as solutions of the problems:
\begin{equation*}
    \min_{\beta^{h,t}} \E^\Q_{t_0}\bigg[\Big(Y^{\CF}_t - \widehat{h_t}(X(t))\Big)^2\bigg],\qquad
    \min_{\beta^{f,t}} \E^\Q_{t_0}\bigg[\Big(Y^{\FV}_t - \widehat{f_t}(X(t))\Big)^2\bigg],
\end{equation*}
with $Y^{\CF}_t$ and $Y_t^{\FV}$ the arguments in the conditional expectations of \eqref{eq: RecursionH} and \eqref{eq: RecursionF} (or \eqref{eq: RecursionV}) of \Cref{prop: RecursiveRelationEPO}, respectively. Detailed pseudo-code for the EPO pricing routine with LSM is given in Algorithm \ref{alg: Algorithm} (see Appendix \ref{app: Algorithm}).

\subsection{EPO replication}
\label{ssec: NumericalReplication}

The task is performed assuming a static pre-commitment strategy \citep[see, e.g.,][]{vigna2020time} that minimizes the exposure of the financial institutions in the two scenarios described by \Cref{eq: MinimizationProblem,eq: MinMaxProblem}, respectively. The first case corresponds to the assumption of a given market price of risk, leading to a standard minimization problem; while the second case aims to find the optimal strategy against the worst possible case.

\subsubsection{Conditional static replication}
\label{sssec: ConditionalReplication}
When a market price of risk $\lambda$ realises, the optimal strategy is computed by solving \Cref{eq: MinimizationProblem} for an objective of the optimization as given, for instance, in \Cref{eq: LossObjectiveMoments,eq: LossObjectiveExpectedShortfall,eq: LossObjectiveGeneral}. 
When a loss function of the kind \eqref{eq: LossObjectiveMoments} is chosen, with order $p=2$, we obtain a convex objective function that allows for an analytical solution. In particular, the objective of the optimization reads:
\begin{equation}
\label{eq: QuadraticLoss}
    \L_{M^2}(\w,\lambda)=\w^\top X \w + \y(\lambda)^\top\w + z(\lambda),
\end{equation}
where $X=\{x_{i,\bar{i}}\}_{i,\bar{i}}\in\R^{I\times I}$, $\y(\lambda)=\{y_i(\lambda)\}_i\in\R^I$, and $z(\lambda)\in\R$.\footnote{By $x_{i,\bar{i}}$ we indicate the value in the $(i,\bar{i})$ position of the matrix $X$, while $y_i(\lambda)$ indicate the $i$-th entry of the vector $\y(\lambda)$.} The coefficients of $x_{i,\bar{i}}$, $y_i(\lambda)$, and $z(\lambda)$ are given by:
\begin{equation}
\label{eq: CoefficientsMeanSqauredError}
    \begin{aligned}
    x_{i,\bar{i}}&=\int_{t_0}^T \alpha^2(t) \E_{t_0}^{\Q_\lambda}\Big[\W_i(t)\W_{\bar{i}}(t)\Big]\d t,\\
    y_i(\lambda)&=\int_{t_0}^T \alpha^2(t)\E_{t_0}^{\Q_\lambda}\Big[\W_i(t)\W_{V}(t;\lambda)\Big]\d t,\\
    z(\lambda)&=\int_{t_0}^T \alpha^2(t)\E_{t_0}^{\Q_\lambda}\Big[\W_{V}^2(t;\lambda)\Big]\d t,
    \end{aligned}
\end{equation}
with $\W_V(t;\lambda)$ and $\W_i(t)$ as in \Cref{eq: WealthHedge,eq: WealthEPO}, and $\alpha(t)$ as in \Cref{eq: LossObjectiveMoments}. We observe that even if the coefficients $x_{i,\Bar{i}}$ depend on $\lambda$ through the reference measure $\Q_\lambda$, they do not depend on $\lambda$ explicitly since the expectation $\E_{t_0}^{\Q_\lambda}\Big[\W_i(t)\W_{\bar{i}}(t)\Big]$ is constant in $\lambda$.

A conditional optimal solution is computed imposing the first-order conditions on the objective function, namely:
\begin{equation}
\label{eq: FirstOrderConditionsMin}
   \nabla_\w\L_{M^2}(\w,\lambda)=0,
\end{equation}
for $\nabla_\w$ the gradient w.r.t. the notional allocated in each instrument. Hence, an optimal strategy $\w^*(\lambda)$ is obtained as:\footnote{$X^{-1}$ is the pseudo-inverse matrix of $X$. If $rank(X)<I$, then the optimal solution is not unique.}
\begin{equation*}
    \w^*(\lambda) = X^{-1}\y(\lambda),
\end{equation*}
where we approximate the coefficients of $X$ and $\y(\lambda)$ in \eqref{eq: CoefficientsMeanSqauredError} using MC sampling for the inner expectations and numerical integration for the outer time-integral.

The theory provided so far only allows us to compute the conditional optimal strategy for a quadratic objective function. We compute the optimal solution via numerical optimization techniques in the general case, such as the one given by the objective in \Cref{eq: LossObjectiveGeneral}.

\subsubsection{Robust static replication}
\label{sssec: RobustReplication}
The robust optimization problem given in \Cref{eq: MinMaxProblem} is a more complex task, even numerically, since it requires a search over the space $\M$. For this purpose, it is convenient to restrict the search domain of the market price of risk $\lambda(t)$. Considering the market price of risk in \eqref{eq: MarketPriceOfRisk}, a first restriction is to consider affine processes in the risk factor $b(t)$, i.e.:
\begin{equation*}
    \lambda(t) = \lambda_0 + \lambda_1 b(t), \qquad \lambda_0,\lambda_1\in\R.
\end{equation*}
This is the specification for $\lambda(t)$ that allows both \eqref{eq: RiskNeutralDynamicsPrepayment} and \eqref{eq: RealWorldDynamicsPrepayment} to be governed by OU dynamics. 
Particularly, the coefficients in \eqref{eq: RiskNeutralDynamicsPrepayment} are computed as:
\begin{equation}
\label{eq: RiskNeutralOUParams}
    \alpha_b^{\Q} = \alpha_b^{\P} + \eta_b\lambda_1, \quad
    \theta_b^{\Q}=\frac{\alpha_b^{\P}\theta_b^{\P} - \eta_b\lambda_0}{\alpha_b^{\Q}},
\end{equation}
with $\alpha^\P_b$, $\theta^\P_b$, and $\eta_b$ from \eqref{eq: RealWorldDynamicsPrepayment}.

Considering the domain restriction $\alpha_b^\Q>0$ for a meaningful OU process, we get the structural constraint $\lambda_1>-{\alpha_b^\P}/{\eta_b}$. Such a restriction requires solving the optimization problem on the semi-plane spanned by $\lambda_1>-{\alpha_b^\P}/{\eta_b}$. Based on the mortgage issuer's risk aversion/appetite, bounds are specified for the search domain. Different choices of such bounds represent different levels of belief regarding the extreme cases to investigate. The risk factor $b$ can be interpreted as a fictitious rate -- a \emph{spread} -- that distorts the people's reaction to the financial rate incentive. Such interpretation is useful to imply a domain restriction on $(\lambda_0,\lambda_1)$ that reflects realistic cases for the parameters $\theta_b^\Q$ and $\alpha_b^\Q$. A similar approach appears in the theory of \emph{good deals} where extreme values for the \emph{Sharpe ratio} were used to imply a realistic family of risk-neutral measures when an incomplete economy was considered \citep[see][]{cochrane2000beyond,bjork2006towards}.

We set $\overline{\alpha_b^\Q}$ as the maximum value for $\alpha_b^\Q$, and $\underline{\theta_b^\Q}$ and $\overline{\theta_b^\Q}$ the minimum and maximum values for $\theta_b^\Q$, respectively. Thus, we define a rectangular domain on $(\theta_b^\Q, \alpha_b^\Q)$, that corresponds to a bounded search domain in $(\lambda_0,\lambda_1)$. In particular, the search domain is identified by the inequalities:
\begin{equation}
\label{eq: SearchDomain}
        \lambda_1 > \underline{\lambda_1},\qquad
        \lambda_1 \leq \overline{\lambda_1},\qquad
        \lambda_0 \geq \underline{m}\lambda_1+\underline{q},\qquad
         \lambda_0 \leq \overline{m}\lambda_1+\overline{q},
\end{equation}
for the constants:
\begin{equation*}
    \underline{\lambda_1}=-\frac{\alpha_b^\P}{\eta_b},\quad \overline{\lambda_1}=\frac{\overline{\alpha_b^\Q}-\alpha_b^\P}{\eta_b},\quad \underline{m}=-\overline{\theta_b^\Q},\quad \overline{m}=-\underline{\theta_b^\Q},\quad \underline{q}=\frac{\alpha_b^\P}{\eta}(\theta_b^\P-\overline{\theta_b^\Q}),\quad \overline{q}=\frac{\alpha_b^\P}{\eta}(\theta_b^\P-\underline{\theta_b^\Q}).
\end{equation*}
For $\underline{\theta_b^\Q}<\overline{\theta_b^\Q}$, the first inequality in \eqref{eq: SearchDomain} is redundant (the last two inequalities imply it). Hence, the ``rectangular domain'' in $(\theta_b^\Q,\alpha_b^\Q)$ coordinates corresponds to a ``triangular domain'' in $(\lambda_0,\lambda_1)$ coordinates with vertexes, respectively:
\begin{equation*}
    \bigg(\frac{\alpha^\P_b\theta^\P_b}{\eta_b},-\frac{\alpha_b^\P}{\eta_b}\bigg),\qquad\bigg(\frac{\alpha^\P_b\theta^\P_b-\overline{\alpha_b^\Q}\underline{\theta_b^\Q}}{\eta_b},\frac{\overline{\alpha_b^\Q}-{\alpha_b^\P}}{\eta_b}\bigg),\qquad\bigg(\frac{\alpha^\P_b\theta^\P_b-\overline{\alpha_b^\Q}\overline{\theta_b^\Q}}{\eta_b},\frac{\overline{\alpha_b^\Q}-{\alpha_b^\P}}{\eta_b}\bigg).
\end{equation*}
Observe that $\theta_b^\P$ controls the horizontal location of the search domain, while the length of the two diagonal sides is increasing in $|\underline{\theta_b^\Q}|$ and $|\overline{\theta_b^\Q}|$, respectively. In particular, the triangle is isosceles when $|\underline{\theta_b^\Q}|=|\overline{\theta_b^\Q}|$.
An illustration is given in \Cref{fig: SearchDomain} for the case $\theta_b^\P=0$, and $\underline{\theta_b^\Q}<0<\overline{\theta_b^\Q}$ with $|\underline{\theta_b^\Q}|<|\overline{\theta_b^\Q}|$.

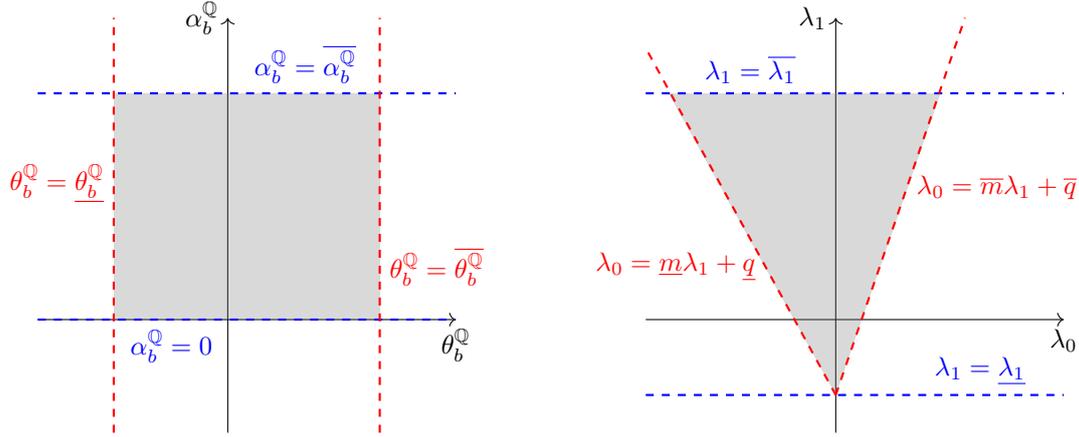
\begin{figure}[t]
    \centering
    \begin{tikzpicture}
    \begin{scope}[xshift=-8cm]
        \coordinate (A2) at (-1.5,0);
        \coordinate (B2) at (2.,0);
        \coordinate (C2) at (2.,3);
        \coordinate (D2) at (-1.5,3);

        \draw[->] (-2.5,0) -- (3,0) node[below] {$\theta_b^{\Q}$};
        \draw[->] (0,-1.5) -- (0,4.) node[left] {$\alpha_b^{\Q}$};

        \fill[gray, opacity=0.3] (A2) -- (B2) -- (C2) -- (D2) -- cycle;

        \draw[blue, dashed, thick] (-2.5,0) -- (3,0) node[pos=0.32, below] {$\alpha_b^\Q=0$};
        \draw[blue, dashed, thick] (-2.5,3) -- (3.,3) node[pos=0.64, above] {$\alpha_b^\Q=\overline{\alpha_b^\Q}$};
        \draw[red, dashed, thick] (-1.5,-1.5) -- (-1.5,4.) node[pos=0.6, left] {$\theta_b^\Q=\underline{\theta_b^\Q}$};
        \draw[red, dashed, thick] (2,-1.5) -- (2,4.) node[pos=0.4, right] {$\theta_b^\Q=\overline{\theta_b^\Q}$};
    \end{scope}

    \coordinate (A1) at (0,-1);
    \coordinate (B1) at (-2.2,3);
    \coordinate (C1) at (1.38,3);

    \draw[->] (-2.5,0) -- (3,0) node[below] {$\lambda_0$};
    \draw[->] (0,-1.5) -- (0,4) node[left] {$\lambda_1$};

    \fill[gray, opacity=0.3] (A1) -- (B1) -- (C1) -- cycle;

    \draw[blue, dashed, thick] (-2.5,-1) -- (3,-1) node[above, pos=0.8] {$\lambda_1=\underline{\lambda_1}$};
    \draw[blue, dashed, thick] (-2.5,3) -- (3,3) node[pos=0.25, above] {$\lambda_1=\overline{\lambda_1}$};
    \draw[red, dashed, thick] (0,-1) -- (-2.5,3.6) node[pos=0.37, left] {$\lambda_0=\underline{m}\lambda_1 + \underline{q}$};
    \draw[red, dashed, thick] (0,-1) -- (1.7,4) node[pos=0.55, right] {$\lambda_0=\overline{m}\lambda_1 + \overline{q}$};

\end{tikzpicture}
    \caption{\footnotesize Domain restriction. Left: Restriction based on the parameters $(\theta_b^\Q, \alpha_b^\Q)$. Right: implied search domain in $(\lambda_0,\lambda_1)$.}
    \label{fig: SearchDomain}
\end{figure}

Once the search domain is selected, the optimization in \Cref{eq: MinMaxProblem} is based on the observation that, given a differentiable objective function, (internal) solutions to the min-max problem $\eqref{eq: MinMaxProblem}$ must satisfy\footnote{This is a necessary, yet not sufficient, condition!} the first order conditions:
\begin{equation}
\begin{aligned}
    \label{eq: FirstOrderConditionsMinMax}
    \nabla_\w \L(\w,\lambda)&=0,\\
    \nabla_\lambda \L(\w,\lambda)&=0,
\end{aligned}
\end{equation}
with $\nabla_\w$ as in \eqref{eq: FirstOrderConditionsMin} and $\nabla_\lambda$ the gradient operator w.r.t. the market price of risk.

We consider the problem with the quadratic loss function in \Cref{eq: QuadraticLoss}. Under this assumption, the first-order conditions in \eqref{eq: FirstOrderConditionsMinMax} are rewritten as:
\begin{align}
    \label{eq: FirstOrderConditionsMinMaxWeights}
    &\w(\lambda)=X^{-1}\y(\lambda),\\
    \label{eq: FirstOrderConditionsMinMaxMarket}
    &\nabla_\lambda z (\lambda) - 2 \big[\nabla_\lambda \y(\lambda)\big] X^{-1} \y(\lambda)=0,
\end{align}
where $\nabla_\lambda \y(\lambda)$ is the Jacobian of $\y(\lambda)$. The solutions of \eqref{eq: FirstOrderConditionsMinMaxMarket}, combined with \eqref{eq: FirstOrderConditionsMinMaxWeights}, provide a domain of search for the actual min-max problem. In particular, saddle-points are the pairs $(\w^*,\lambda^*)$ where $\L$ is convex in $\w$ and concave in $\lambda$. This can be checked by inspecting the Hessian matrix of $\L$. The above routine allows for identifying the saddle-points within the domain of interest. An analysis of the search domain boundary is required to obtain all the significant solutions of \eqref{eq: MinMaxProblem} on a bounded domain.

\section{Numerical experiments}
\label{sec: NumericalResults}

\begin{table}[b]
\begin{minipage}{0.45\textwidth}
\centering
\caption{Parameters for SDEs in \Cref{eq: RiskNeutralDynamicsRate} and \Cref{eq: RealWorldDynamicsPrepayment}, respectively.}
\begin{tabular}{c||c|c|c||c}
\toprule
\multirow{2}{*}{$r$} & $\alpha_r^{\Q}$ & $\vartheta_r^{\Q}(t)$ & $\sigma_r$ & \multirow{2}{*}{$\rho$}\\
&  $0.023$ & $YC\equiv3\%$ & 0.006 & \\
\cmidrule(lr){1-5}
\multirow{2}{*}{$b$}  & $\alpha_b^{\P}$ & $\theta_b^{\P}$ & $\sigma_b$ & \multirow{2}{*}{$0.44$}\\
& $2.099$ & $-0.002$ & $0.015$ & \\
\bottomrule
\end{tabular}

\label{tab: SDEsParameters}
\end{minipage}
\hspace{0.05\textwidth} 
\begin{minipage}{0.45\textwidth}
\centering
\caption{Parameters for the sigmoid rate incentive function in \Cref{eq: CPRSigmoid}.}
\begin{tabular}{c||c|c|c}
\toprule
$h_{RI}$ & $l$ & $u$ & $a$ \\
\cmidrule(lr){1-4}
empirical & $0.0231$ & $0.0447$ & $84$ \\
rational & $0.0$ & $0.0447$ & $+\infty$ \\
\bottomrule
\end{tabular}
\label{tab: SigmoidParameters}
\end{minipage}
\end{table}

In this section, we report some numerical experiments based on the model described in \Cref{sec: PrepaymentModel}, implemented according to the methodology and assumptions of \Cref{sec: NumericalPricingReplication}.

The experiments are run considering the following setup. For \Cref{eq: RiskNeutralDynamicsRate}, we set $\alpha_r^\Q=0.023$ and $\sigma_r=0.006$. These choices correspond to realistic coefficients for the Hull-White model, obtained from calibration on a market-realized swaption volatility cube. $\vartheta_r^\Q(t)$, defined as in \citep{hull1990pricing} so to recover today's market zero-coupon bond curve, takes as input a flat yield curve, $YC$, to the level $YC\equiv 3\%$). The coefficients in \Cref{eq: RealWorldDynamicsPrepayment} are computed with MLE and are based on a time series for $b(t)$ calibrated on realistic prepayment data (see \Cref{fig: Sigmoid}). The coefficients are $\alpha_b^\P=2.099$, $\theta_b^\P=-0.002$ and $\eta_b=0.015$. The correlation coefficient equals $\rho=0.44$ (see \Cref{tab: SDEsParameters}). 
The structural parameters of the sigmoid rate incentive function, $h_{RI}$, in \Cref{eq: CPRSigmoid} are also calibrated based on the data represented in \Cref{fig: Sigmoid} and set to $l=0.0231$, $u=0.0447$, and $a=84$. For the sake of comparison, we also consider a rational rate incentive obtained by setting the lower bound $l=0.0$ and the steepness parameter $a\to+\infty$ (see \Cref{tab: SigmoidParameters}). 

In \Cref{sssec: CheapPriceEPO,ssec: ExposureReplicationSimple}, the market price of risk for $b(t)$ is assumed to be zero, $\lambda(t)=0$. In other words, we perform the pricing experiment under the assumption that the risk-neutral dynamics in \Cref{eq: RiskNeutralDynamicsPrepayment} coincide with the real-world dynamics in \Cref{eq: RealWorldDynamicsPrepayment}, as e.g. is done in \citep{aid2009structural} in the context of the energy market. The effect of different choices of the market price of risk is investigated in \Cref{sssec: EffectLambda,ssec: ExposureReplicationRobust}.

For reference, the initial notional of the mortgage contract is set to $N_{\tt{c},0}=10^4$, and hence the EPO value is reported and visualized in basis points (bps) of the initial notional. Furthermore, we consider contracts with tenor from $0$ to $10$ years, with one single payment per year and a fixed rate of $3.1\%$. The amortization scheme used in the experiments is either interest-only (also said bullet), or linear. The details are reported in \Cref{tab: EPOSpecs}.

\begin{table}[t]
\begin{minipage}{1.\textwidth}
\centering
\caption{Mortgage contract specification.}
\begin{tabular}{c|c|c|c|c}
\toprule
notional, $N_{\tt{c},0}$ & fixed rate, $K$ & tenor, $(t_m,t_n)$ & freq. (yearly) & amortization\\
\cmidrule(lr){1-5}
$10^4$ & $3.1\%$ & $(0, 10)$& $1$ & bullet/linear\\
\bottomrule
\end{tabular}
\label{tab: EPOSpecs}
\end{minipage}
\end{table}

\subsection{Pricing of the EPO}
\subsubsection{Effect of uncertainty on the EPO price}
\label{sssec: CheapPriceEPO}

In this experiment, we want to test the statement in \Cref{thm: DecreasingPrepaymentPrice}. Furthermore, we observe similar qualitative behavior in more general cases (even if it is not guaranteed by \Cref{thm: DecreasingPrepaymentPrice}).

For illustration purposes, we consider two mortgage contracts: a 10-year interest-only (bullet) and a 10-year linear mortgage with annual interest payments and notional repayments. The number of payment dates is $n=10$, and they fall at times $t_j=j$ years, for $j=1,\dots,n$. We observe the effect on the EPO value when considering both a rational rate incentive function and an empirical one with parameters given in \Cref{tab: SigmoidParameters}. The two RI functions are illustrated in \Cref{fig: Sigmoid}a: dashed red and solid blue lines represent the rational and empirical RI functions, respectively. Empirical evidence shows a systematic minimum prepayment level, i.e. $l>0$, even when the rate incentive $\varepsilon$ is negative. This is a major difference between the empirically calibrated sigmoid and the rational rate incentive function.

First, the claims of \Cref{thm: DecreasingPrepaymentPrice} and \Cref{cor: CheapPrepaymentPrice} are confirmed numerically for a bullet contract, setting $\rho=0$ and $\theta_b^\Q=0$.
In \Cref{fig: CheapPrice}a, the solid black line represents the price of the EPO for a bullet mortgage with a rational rate incentive function and parameters $\rho=0$ and $\theta_b^\Q=0$. The line is decreasing in $\eta_b$ and maximum for $\eta_b=0$. The case $\eta_b = 0$ corresponds to the model introduced by \citep{casamassima2022pricing}.
A similar qualitative behavior is observed when the realistic $\rho$ and $\theta_b^\Q$ are used, even for different amortization schemes. Dotted-solid blue and dotted-dashed red lines in \Cref{fig: CheapPrice}a are the EPO values for bullet and linear amortization mortgage, respectively. The effect of prepayment is less significant the faster the contractual amortization scheme is, hence, the value of the prepayment option is lower for a linear mortgage than for a bullet.

\begin{figure}[b!]
    \centering
    \subfloat[\centering]{{\includegraphics[width=7.5cm]{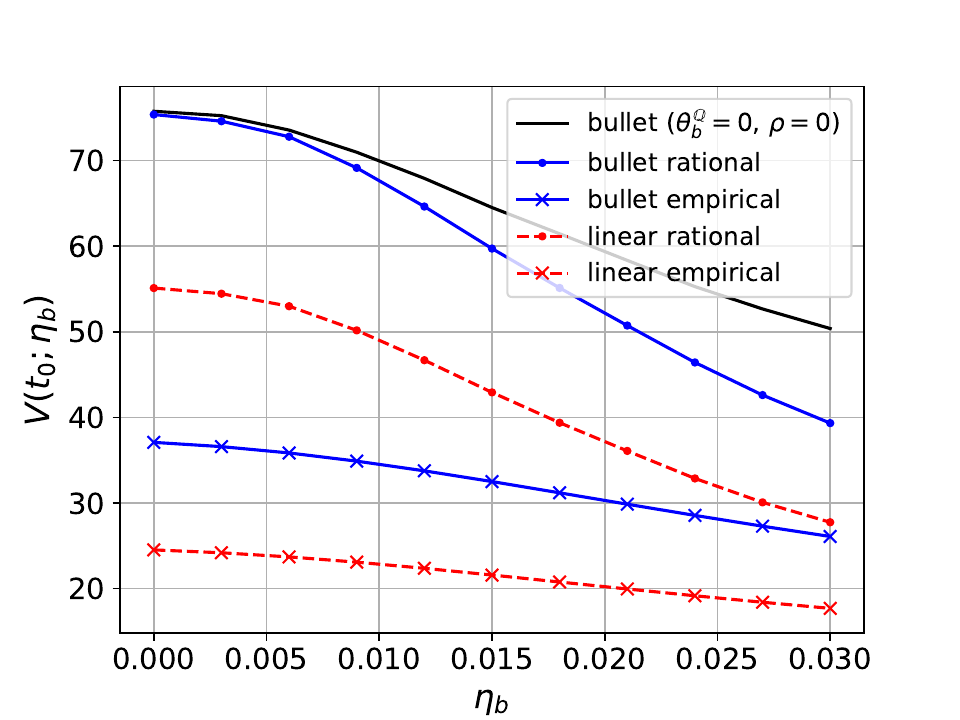} }}
    ~\hspace{-.5cm}
    \subfloat[\centering]{{\includegraphics[width=7.5cm]{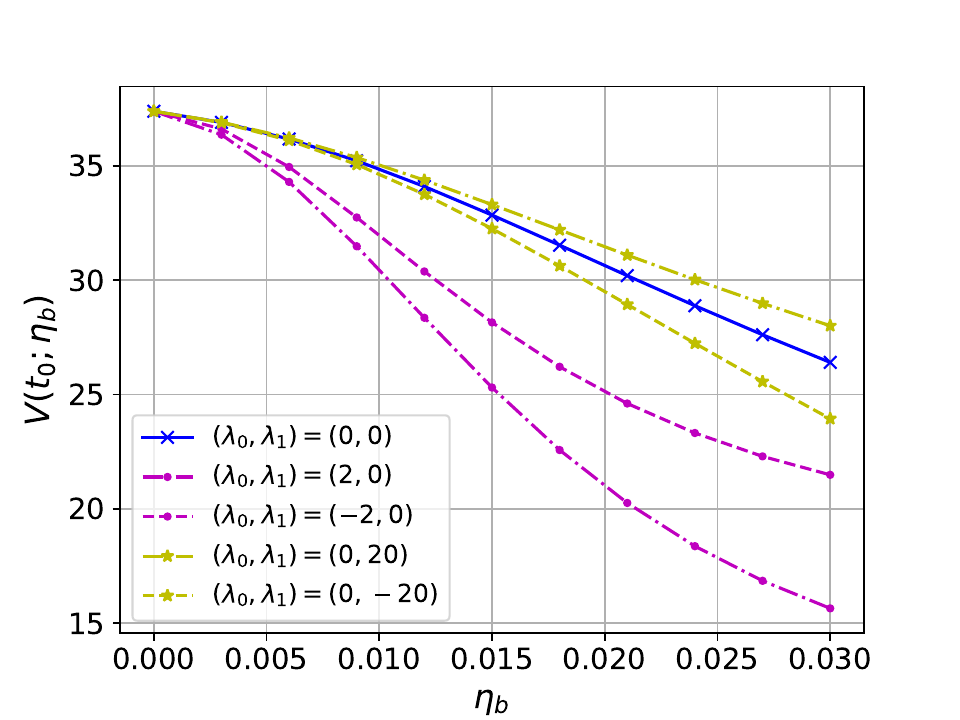} }}%
    \caption{(a) EPO value $V(t_0)\equiv V(t_0;\eta_b)$ plotted against the volatility parameter $\eta_b$ for different amortization schemes and rate incentive functions. (b) EPO value $V(t_0)\equiv V(t_0;\eta_b)$ for a bullet mortgage with empirical rate incentive function for different assumptions of the market price of risk $(\lambda_0,\lambda_1)$.}
    \label{fig: CheapPrice}%
\end{figure}

The experiment is repeated for the empirically calibrated sigmoid function. The qualitative behavior is consistent with the previous experiments. 
However, the reduced magnitude in the EPO value entails a less evident effect of uncertainty.
We observe two causes. On the one hand, a recurrent non-zero prepayment for $\varepsilon < 0$ (see \eqref{eq: RateIncentive}) entails that some prepayments are exercised even when it is inconvenient to prepay; on the other hand, the continuous transition from a negative to a positive RI represents a delayed prepayment even when a positive $\varepsilon$ is realized. 
In \Cref{fig: CheapPrice}a, the EPO values for bullet and linear contracts are represented by crossed-solid blue and crossed-dashed red lines, respectively.

\subsubsection{Effect of the market price of risk on the EPO price}
\label{sssec: EffectLambda}
In \Cref{fig: CheapPrice}b, we present the effect of different assumptions of the market price of risk $\lambda(t)=\lambda_0 + \lambda_1 b(t)$ on the EPO value of a 10-year bullet mortgage when the empirical rate incentive function is employed. We compare the results against the solid-crossed blue line (also in \Cref{fig: CheapPrice}a) representing the case with $\lambda_0=\lambda_1=0$. From a qualitative level, variations of $\lambda_0$ (magenta lines in \Cref{fig: CheapPrice}b) give rise to a reduction of the EPO value. 
Indeed, $\lambda_0\neq 0$ entails a long-term horizontal shift of the rate incentive function, representing a suboptimal prepayment exercise. The effect of $\lambda_1$ on the EPO value is different and depends on the sign of $\lambda_1$. Observing that the mean-reversion rate, $\alpha_b^\Q$, of $b(t)$ is monotonically increasing in $\lambda_1$ (see \eqref{eq: RiskNeutralOUParams}) and has an opposite effect to the diffusion parameter, $\eta_b$, the EPO value is monotonically increasing in $\lambda_1$. Confirmation of this fact is given by the yellow lines in \Cref{fig: CheapPrice}b.

It is worth noticing that the sensitivity of the EPO value to the market price of risk exposes any financial institution to the risk of model misspecification whenever a pricing task is required. As we illustrated above, an incorrect assumption of the market price of risk may lead to a significant mispricing error, that is amplified by a higher volatility, $\eta_b$. In \Cref{ssec: ExposureReplicationRobust}, we show an example of a replication strategy that is less sensitive to the market price of risk misspecification.

\subsection{Conditional exposure replication}
\label{ssec: ExposureReplicationSimple}

In this section, we show how the EPO exposure of a 10-year bullet mortgage with yearly interest payments is replicated when the market price of risk is assumed to be known, and we investigate the effect of using different tradable instruments for this purpose. In this experiment, the parameters of the risk factor dynamics are fixed as in \Cref{tab: SDEsParameters}, while the sigmoid parameters are given in \Cref{tab: SigmoidParameters}. The market price of risk is assumed $\lambda(t)=0$ for every $t\in\T$.

For the sake of illustration, we report the results using as a metric the \emph{integrated distance} for the allocation strategy $\w$, which reads:
\begin{equation*}
    \ID(\w)=\int_{t_0}^{t_n} \D(t;\w) \d t,
\end{equation*}
for the distance $\D$, defined in \Cref{eq: SignedDistance}. Different choices of allocation strategies will be reported as arguments of $\ID$. In particular, $\ID(\0)$ is the integrated distance between the value in EPO and a zero hedge, i.e. it is the integrated EPO exposure. 
\begin{table}[b]
\begin{minipage}{1.\textwidth}
\centering
\caption{Replicating instruments specification.}
\begin{tabular}{c||c|c|c|c}
\toprule
instrument & fixed rate, $K$ & tenor, $(t_m,t_n)$ & freq. (yearly) & maturity\\
\cmidrule(lr){1-5}
rec. swap & $3.0\%$ (par) & $(0, 10)$& $1$ & /\\
rec. swaption & $3.0\%$ (ATM) & $(9, 10)$& $1$ & $9$\\
pay. swaption & $3.0\%$ (ATM) & $(9, 10)$& $1$ & $9$\\
\bottomrule
\end{tabular}
\label{tab: HedgeSpecs}
\end{minipage}
\end{table}

\subsubsection{Nonlinear replication of the EPO exposure}
\begin{table}[t]
\begin{minipage}{1.0\textwidth}
\centering
\caption{Optimal allocation and mean squared error loss for different replicating portfolios.}
\begin{tabular}{c||c|c|c|c|c|c|c|c}
\toprule
& \multicolumn{8}{c}{replicating portfolio}\\
& $\0$ & $\w_{L2,1}^*$ & $\w_{L2,2}^*$ & $\w_{L2,3}^*$ & $\w_{L2,4}^*$ & $\w_{L2,5}^*$ & $\w_{L2,6}^*$ & $\w_{L2,7}^*$\\
\cmidrule(lr){1-9}
rec. swap & / & 2066 & / & / & 1677 & 2326 & / & 1528 \\
rec. swaption & / & / & 15180 & / & 5970 & / & 16747 & 6976\\
pay. swaption & / & / & / & -8225 & / & 3857 & -10513 & -1244 \\
\cmidrule(lr){1-9}
$\L_{M^2}$ & 267830 & 19616 & 108251 & 218174 & 3486 & 11838 & 24868 & 3138 \\
$\L_{M^2}$ $(\%)$ & 100\% & 7.32\% & 40.42\% & 81.46\% & 1.30\% & 4.42\% & 9.28\% & 1.17\% \\
\cmidrule(lr){1-9}
initial cost & / & 0 & 75 & -41 & 29 & 19 & 31 & 28 \\
\bottomrule
\end{tabular}
\label{tab: OptimalReplicatingPortfolios}
\end{minipage}
\end{table}

In this experiment, the objective function is of the type given in \Cref{eq: LossObjectiveMoments}, with $p=2$, hence it is used to minimize the mean squared distance between the EPO and the hedge.
We test the optimal replication for different choices of replicating instruments. Particularly, we consider a 10-year receiver swap at par (fixed rate $K=3\%$), a 9-year maturity 1-year tenor receiver swaption, and a 9-year maturity 1-year tenor payer swaption, both at the money (strike $K=3\%$). The detailed specifications for the hedging instruments are given in \Cref{tab: HedgeSpecs}. We indicate the optimal allocation with $\w_{L2,i}^*$, where $i$ indicates $i$-th subset of instruments used in the replication. For instance, $i=1$ indicates the replicating portfolio composed by the swap only, while $i=4$ indicates the portfolio composed by the swap and the receiver swaption. With $\0$, we indicate the trivial no-action strategy, where no instruments are bought or sold for hedging purposes. All the labels for the different replicating portfolios are given in \Cref{tab: OptimalReplicatingPortfolios}. \Cref{tab: OptimalReplicatingPortfolios} also reports the notional invested in each of the instruments, the absolute and relative losses, and the initial cost of hedging. The relative loss is defined as a percentage of the loss given no action, i.e. the absolute loss under the strategy $\0$.

As reported in \Cref{tab: OptimalReplicatingPortfolios}, the strategy $\w_{L2,1}^*$ with a long position in the receiver swap reduces $\L_{M^2}$ of more than $90\%$. However, the skew of the EPO exposure (red histogram in \Cref{fig: ExposureDistribution}a) is more pronounced (towards the right tail) than the swap case. In the exposure of the hedged position, we observe a fat right tail (green histogram). The replication based on solely the receiver or the payer swaption (resp. strategies $\w_{L2,2}^*$ and $\w_{L2,3}^*$ in \Cref{tab: OptimalReplicatingPortfolios}) performs rather poorly since the nonlinear instruments can only capture either the right tail (yellow histogram in \Cref{fig: ExposureDistribution}b) or the left tail (cyan histogram) of the exposure, respectively.

\begin{figure}[b]
\centering
    \subfloat[\centering]{{\includegraphics[width=7.5cm]{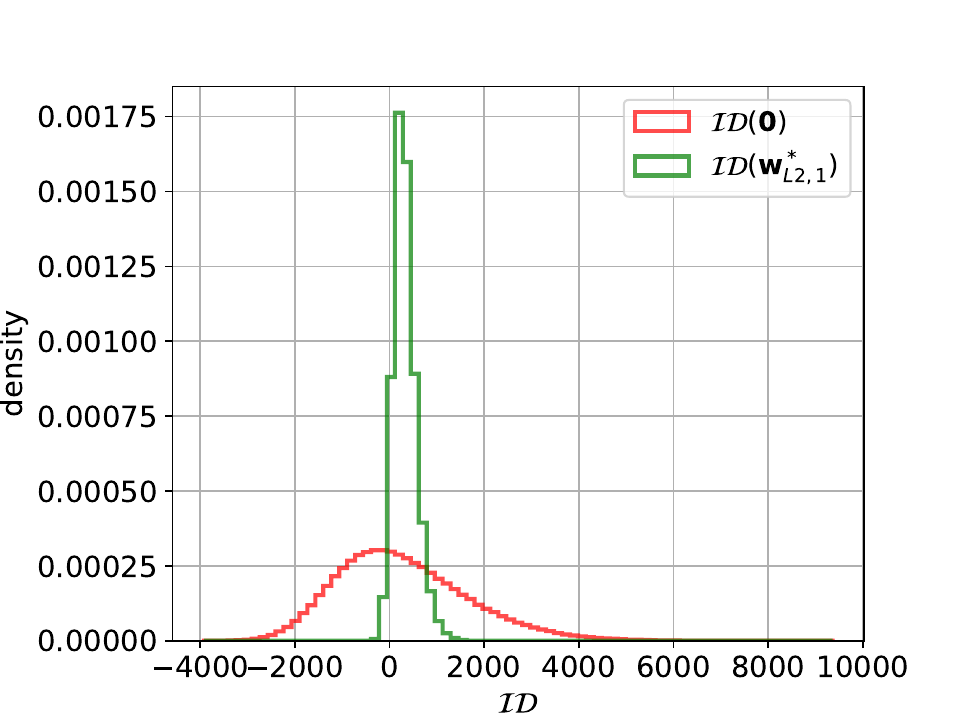} }}%
    ~\hspace{-.6cm}
    \subfloat[\centering]{{\includegraphics[width=7.5cm]{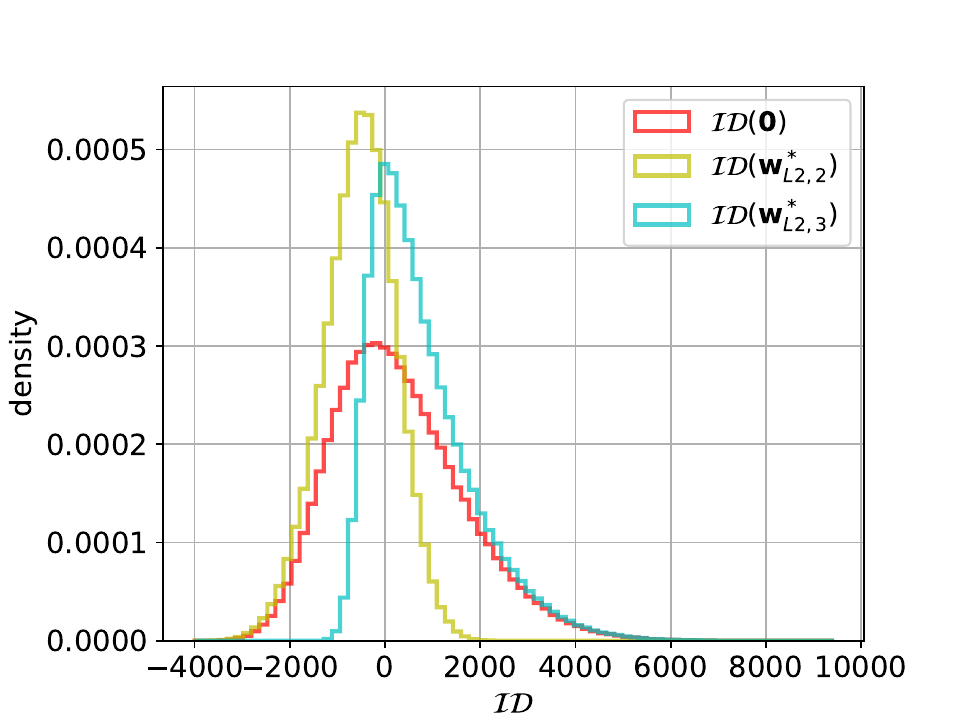} }}
    \caption{Integrated EPO exposure distribution, $\ID(\0)$, in red, compared with the integrated distance, $\ID(\cdot)$, for different hedging strategies: (a) $\w^*_{L2,1}$, in green; (b) $\w^*_{L2,2}$ in yellow and $\w^*_{L2,3}$ in cyan.}%
    \label{fig: ExposureDistribution}%
    \vspace{-0.3cm}
\end{figure}
\begin{figure}[t!]
    \centering
    \subfloat[\centering]{{\includegraphics[width=7.5cm]{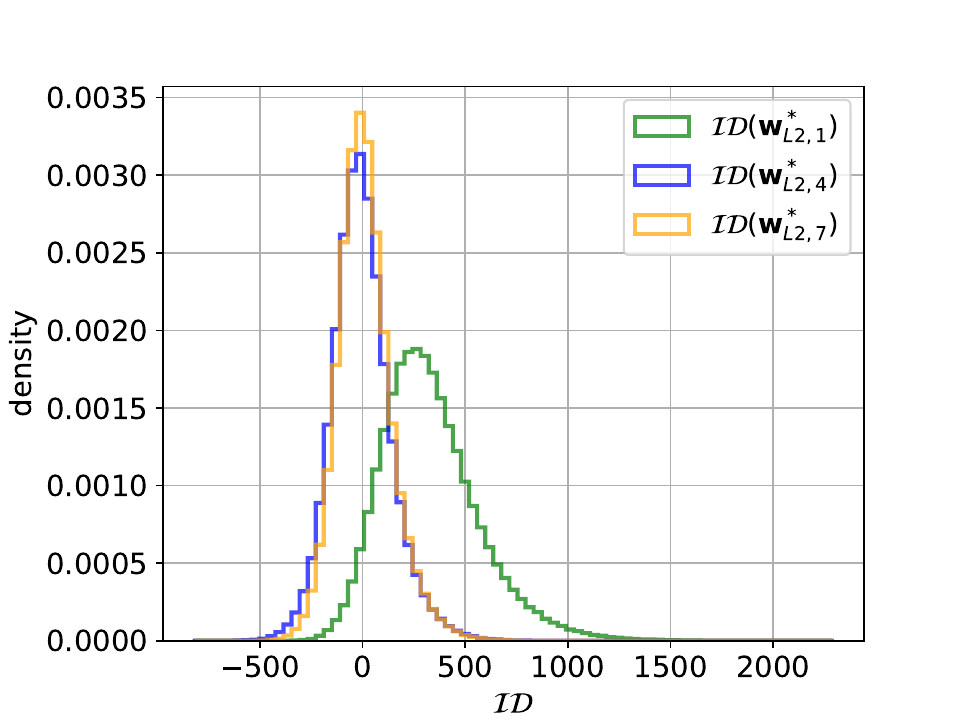} }}%
    ~\hspace{-.6cm}
    \subfloat[\centering]{{\includegraphics[width=7.5cm]{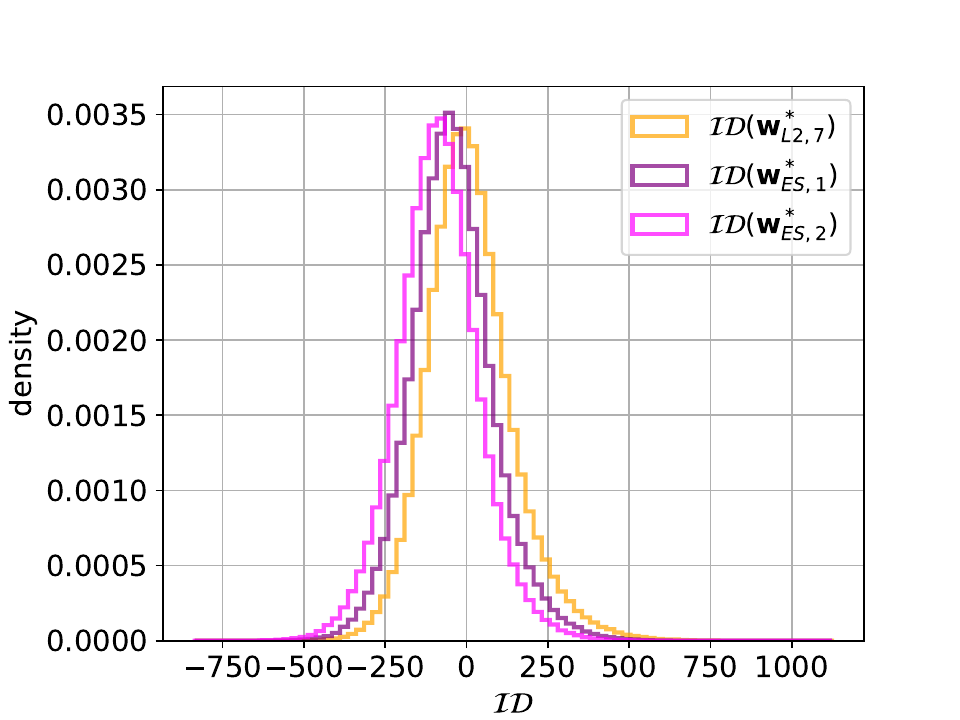} }}
    \caption{Integrated distance distribution, $\ID$, for different strategies and objective functions. (a) Loss function as in \Cref{eq: LossObjectiveMoments} and strategies $\w_{L2,1}^*$ (in green),  $\w_{L2,4}^*$ (in blue), and  $\w_{L2,7}^*$ (in yellow). (b) Loss function as in \Cref{eq: LossObjectiveGeneral}, for $k=0$ ($\w_{L2,7}^*$, in yellow), $k=10$ ($\w_{ES,1}^*$, in purple), $k=20$ ($\w_{ES,2}^*$, in magenta).}
    \label{fig: ExposureDistributionMultipleInstruments}%
\end{figure}

\Cref{tab: OptimalReplicatingPortfolios} shows that when the nonlinear instruments are combined with the swap ($\w_{L2,4}^*$ and $\w_{L2,5}^*$) the replication improves because of the asymmetry in the swaptions' exposure, with the best performance when the receiver swaption is included in the strategy. Using only the nonlinear instruments, as in $\w_{L2,6}^*$, leads to an inaccurate replication, confirming that an important component of the EPO exposure resembles a linear instrument. A huge notional is invested in the two swaptions (with opposite signs) achieving an effect similar to a receiver swap, using a put-call parity argument.

The best replication is obtained using the swap and the swaptions (in the strategy $\w_{L2,7}^*$). The positive notional invested in the receiver swap represents the linear component in the EPO, the positive notional invested in the receiver swaption and the negative notional invested in the payer swaption help control the exposure's tails and skew. The different magnitudes indicate that most of the optionality resembles a receiver swaption. In \Cref{fig: ExposureDistributionMultipleInstruments}a, we show that $\w_{L2,4}^*$ (blue histogram) and $\w_{L2,7}^*$ (orange histogram) significantly reduce the right tail of the exposure, compared to $\w_{L2,1}^*$ (green histogram). The exposure distribution resulting from $\w_{L2,7}^*$ is slightly more peaked compared to $\w_{L2,4}^*$ since the additional payer swaption reduces the left tail of the EPO exposure.

In general, we observe positive notional investment in the receiver swap and swaption, while negative notional for the payer swaption. The long linear instrument is used to hedge the center of the EPO exposure, whereas the swaptions are used to hedge the tails. The more pronounced right tail in the EPO exposure entails a higher notional in the receiver swaption compared to the payer swaption. The only exception is $\w_{L2,5}^*$. Here, the right tail of the exposure is replicated by the receiver swap, while the effect of the payer swaption is negligible. The left tail of the exposure, on the contrary, is obtained as the combination of the swap left tail, which is ``too negative,'' and the payer swaption right tail, which is positive and compensates for that.

From a practical perspective, we can use the nonlinear hedging strategies computed above to significantly improve the EPO hedge. For instance, with an investment of 28 bps of the mortgage portfolio notional in swaptions, we can reduce the cumulated exposure as defined in \eqref{eq: LossObjectiveMoments} more than six times compared to the hedge based on sole swaps (see \Cref{tab: OptimalReplicatingPortfolios}).

\subsubsection{Tail replication of the EPO exposure}

\begin{table}[b]
\begin{minipage}{1.\textwidth}
\centering
\caption{Optimal allocation for loss function including the expected shortfall to the 90\% level, for different choices of $k=0,10,20$.}
\begin{tabular}{c||c|c|c}
\toprule
& \multicolumn{3}{c}{replicating portfolio}\\
 & $\w_{L2,7}^*$ $(k=0)$  & $\w_{ES,1}^*$ $(k=10)$ & $\w_{ES,2}^*$ $(k=20)$\\
\cmidrule(lr){1-4}
rec. swap & 1528 & 1456 & 1427\\
rec. swaption & 6976 & 7928 & 8522 \\
pay. swaption & -1244 & -1242 & -1050 \\
\cmidrule(lr){1-4}
$\L_{M^2}$ & 3138 & 3456 & 4285 \\
$\L_{ES^+_{0.9}}$ & 277 & 201 & 145 \\
\cmidrule{1-4}
initial cost & 28 & 33 & 37\\
\bottomrule
\end{tabular}
\label{tab: OptimalReplicatingPortfoliosES}
\end{minipage}
\end{table}

In this experiment, we test a different loss function, of the kind given in \Cref{eq: LossObjectiveGeneral} with the expected shortfall to the 90\% level, for different choices of the tuning parameter $k$. The hedging instruments are selected as in the previous experiment. The optimal allocation is indicated by $\w_{ES,1}^*$ and $\w_{ES,2}^*$, for $k=10,20$, respectively. The results are benchmarked with $\w_{L2,7}^*$ which corresponds to the optimal strategy obtained setting $k=0$. In \Cref{tab: OptimalReplicatingPortfoliosES}, we report the different optimal strategies. The allocation in the receiver swap decreases when $k$ increases. However, the change in allocation is more pronounced for the nonlinear instruments. The notional invested in the receiver (resp. payer) swaption increases (resp. decreases) for an increasing $k$. This is consistent with the expectation that nonlinear instruments are mainly responsible for the replication in the tails of the EPO exposure.

From \Cref{fig: ExposureDistributionMultipleInstruments}b and \Cref{tab: OptimalReplicatingPortfoliosES}, we observe that including the expected shortfall in the objective shifts the distribution to the left, allowing the control of extreme scenarios. The potential losses are reduced, but the hedging cost increases, when $k$ increases. Parameter $k$ needs to be tuned to achieve satisfactory control on the right tail and a tolerable price for the hedging. 
For instance, considering as benchmark strategy $\w_{L2,7}$ and investing 5 additional bps (about 18\% of $\w^*_{L2,7}$ value) to achieve $\w^*_{ES,1}$, we may reduce $\L_{ES^+_{0.9}}$ -- a measure of the cumulated expected shortfall over time (see \Cref{eq: LossObjectiveExpectedShortfall}) -- by more than 27\%, obtaining a robust control on the right tail of the exposure. Besides the additional cost, however, $\w^*_{ES,1}$ performs less accurately in the middle and left tail of the exposure, as we deduce by an increase of about 10\% of $\L_{M^2}$. A similar argument holds for $\w^*_{ES,2}$. The results are reported in \Cref{tab: OptimalReplicatingPortfoliosES}.

\subsection{Unconditional exposure replication}
\label{ssec: ExposureReplicationRobust}

This section is dedicated to studying the effect of different choices of the market price of risk, $\lambda$, in the exposure replication (see \Cref{sssec: RobustReplication}). For the sake of this experiment, we consider a 10-year bullet mortgage with yearly interest payments. The parameters of the risk factors dynamics are given in \Cref{tab: SDEsParameters}, and those of the empirical sigmoid are reported in \Cref{tab: SigmoidParameters}.

\begin{figure}[b!]
    \centering
    \includegraphics[width=1.\textwidth]{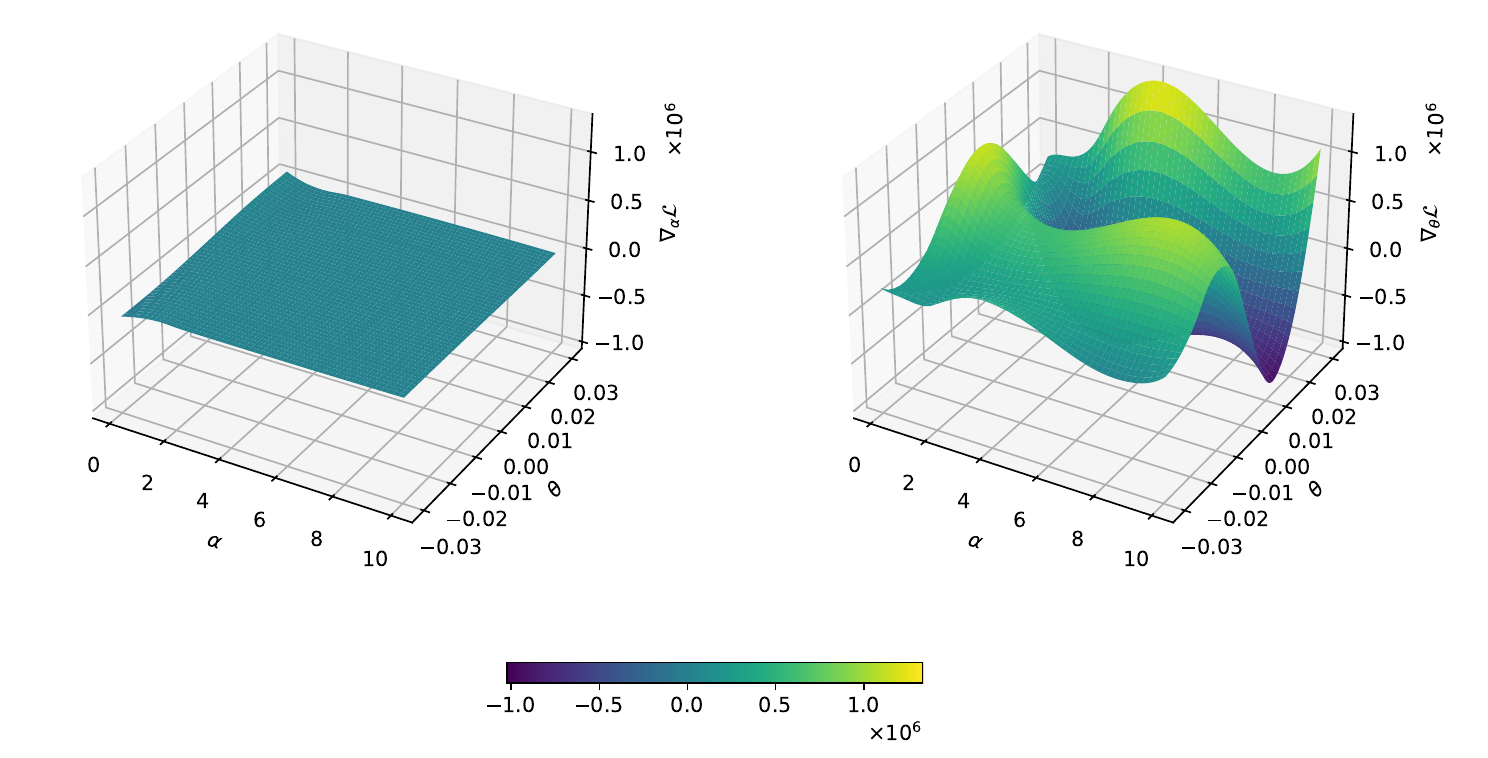}
    \caption{Gradient of the loss function on the subspace spanned by \eqref{eq: FirstOrderConditionsMinMaxWeights}. Left: partial derivative w.r.t. $\alpha$. Right: partial derivative w.r.t. $\theta$.}
    \label{fig: Gradient}
\end{figure}

\begin{figure}[t!]
    \centering
    \includegraphics[width=1.\textwidth]{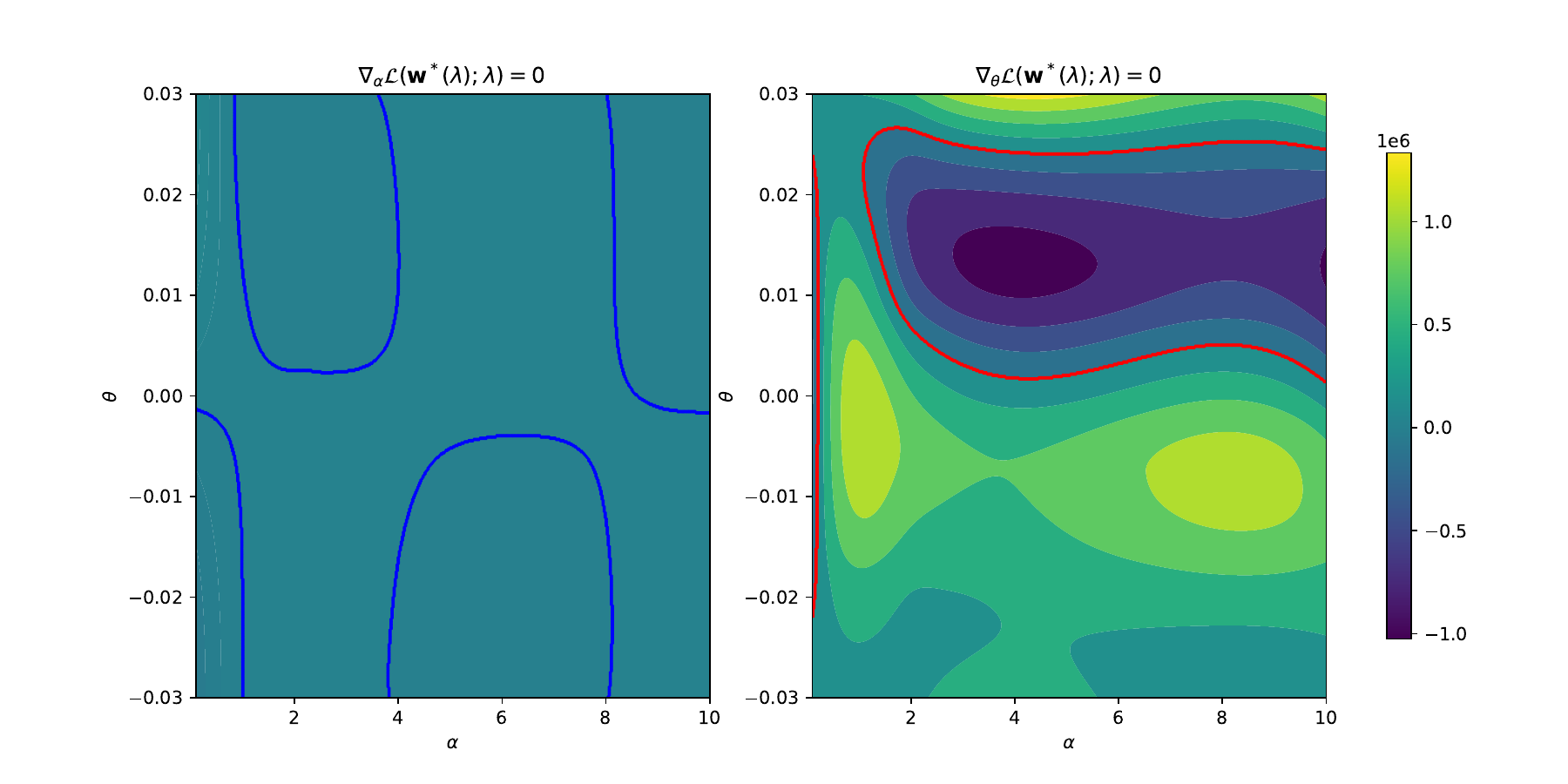}
    \caption{Contour plot of the gradient in \Cref{fig: Gradient}. The set $\{\lambda:\:\nabla_\alpha \L(\w^*(\lambda);\lambda)=0\}$ is represented with a blue line. The set $\{\lambda:\:\nabla_\theta \L(\w^*(\lambda);\lambda)=0\}$ is represented with a red line.}
    \label{fig: GradientProjection}
\end{figure}

Starting from the $\P$-dynamics for $b$ in \Cref{tab: SDEsParameters}, we set the bounds on the search region for the market price of risk implied by the parameter ranges $\alpha_b^\Q\in[0.1,10.0]$ and $\theta_b^\Q\in [-0.03, 0.03]$. 
The loss function considered is defined in \Cref{eq: LossObjectiveMoments}, for $p=2$. The hedging instrument is the receiver swap from \Cref{tab: HedgeSpecs}.

Notation-wise, $\alpha_b$ and $\theta_b$ are just indicated with $\alpha$ and $\theta$, in equations and figures. Furthermore, the market price of risk can be parameterized in $\alpha_b$ and $\theta_b$: we will refer to the pair $(\alpha_b, \theta_b)\equiv(\alpha,\theta)$ as ``the market price of risk,'' and we will indicate the partial derivatives ``w.r.t. the market price of risk'' with $\nabla_\alpha$ and $\nabla_\theta$.

From \eqref{eq: FirstOrderConditionsMinMaxWeights}, observing that $X^{-1}$ is positive definite independently of $\lambda$, we deduce that the loss function is convex along the direction spanned by the allocation variable. In other words, $X^{-1}\y(\lambda)$ is the optimal conditional strategy $\w^*(\lambda)$. As a necessary condition, the candidate saddle-points must satisfy \eqref{eq: FirstOrderConditionsMinMaxMarket}. Since $\y(\alpha,\theta)$ and $z(\alpha,\theta)$ are not explicitly known, we compute them on a set of nodal points via Monte Carlo simulation. The grid of nodal points is interpolated by a $(3,3)$-degree bivariate spline, that is employed to compute the numerical solution of \eqref{eq: FirstOrderConditionsMinMaxMarket}.

\begin{table}[b]
\begin{minipage}{1.\textwidth}
\centering
\caption{Solutions of the robust hedging problem.}
\begin{tabular}{c||c|c|c|c}
\toprule
$(\alpha,\theta)$ & $(8.25,0.005)$ & $(8.10,0.03)$ & $(0.90,0.03)$ & $(0.10,-0.025)$\\
saddle & YES & NO & NO & NO\\
\cmidrule(lr){1-5}
rec. swap & $2182$& $2526$ & $2425$ & $1880$\\
$\L_{M^2}$ & $24623$ & $17363$ &$31525$ & $18423$\\
\bottomrule
\end{tabular}
\label{tab: RobustSolution}
\end{minipage}
\end{table}

\Cref{fig: Gradient} illustrates the two components of the gradient of $\L$ projected on the subspace spanned by \eqref{eq: FirstOrderConditionsMinMaxWeights}, i.e. the gradient assuming the optimal allocation is always ensured. The partial derivative in $\theta$, $\nabla_\theta\L$, shows a clear change in its sign and attains huge -- positive and negative -- values, while the partial derivative in $\alpha$, $\nabla_\alpha\L$, has a much narrower range around zero. So, a change in value along $\theta$ may be more relevant than along $\alpha$.
The solution to \Cref{eq: FirstOrderConditionsMinMaxMarket} is represented in \Cref{fig: GradientProjection}. Blue and red lines indicate the sets $\nabla_\alpha\L(\w^*(\lambda);\lambda)=0$ and $\nabla_\theta\L(\w^*(\lambda);\lambda)=0$, respectively. Inspecting the Hessian matrix at the points fulfilling both first-order conditions, only $(\alpha,\theta)=(8.25,0.005)$ is an admissible solution, i.e. a saddle-point. We search for other solutions on the search domain boundary and report them in \Cref{tab: RobustSolution}. The four pairs of market prices of risk and allocations in \Cref{tab: RobustSolution} are robust solutions to the hedging problem. In these points, any small market price of risk variation (within the search domain) entails reducing exposure. In \Cref{fig: Trajectories}, we observe that for every initial belief regarding the market price of risk (black dots), one of the four solutions in \Cref{tab: RobustSolution} (black `x's) corresponds to the optimal robust hedge that ensures an upper bound on the exposure. Such upper bound is the loss reported in \Cref{tab: RobustSolution}.

\begin{figure}[t]
    \centering
    \hspace{0.0cm}
    \subfloat[\centering]{{\includegraphics[width=6.6cm]{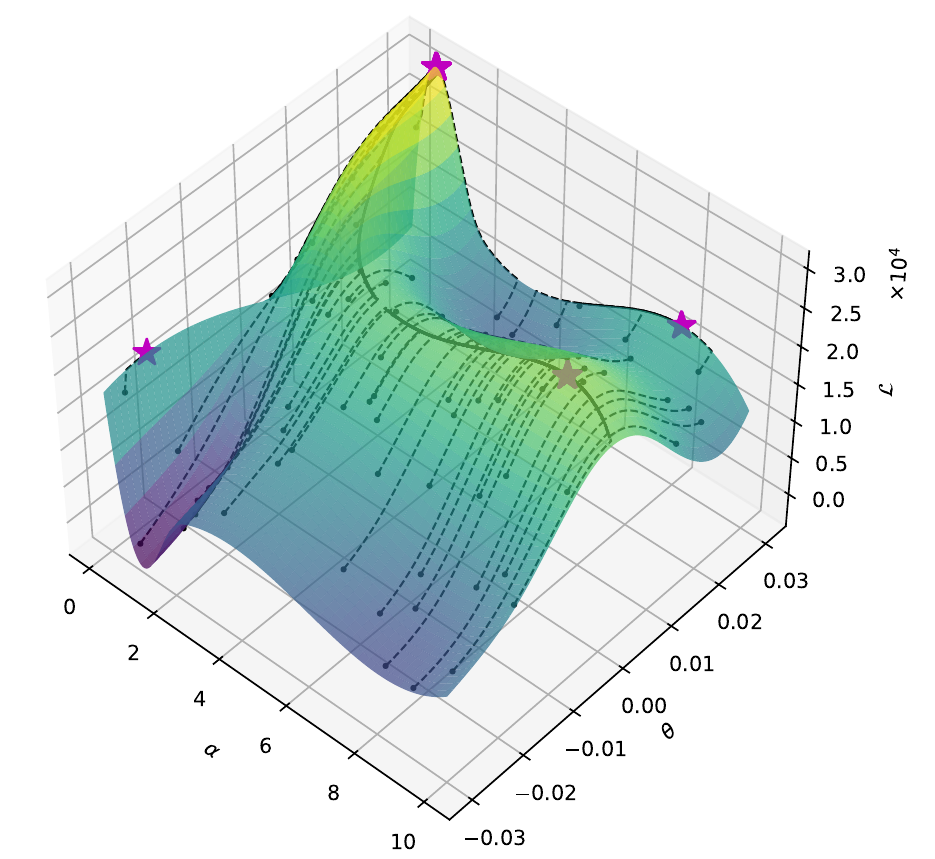} }}%
    ~\hspace{0.0cm}
    \subfloat[\centering]{{\includegraphics[width=7.2cm]{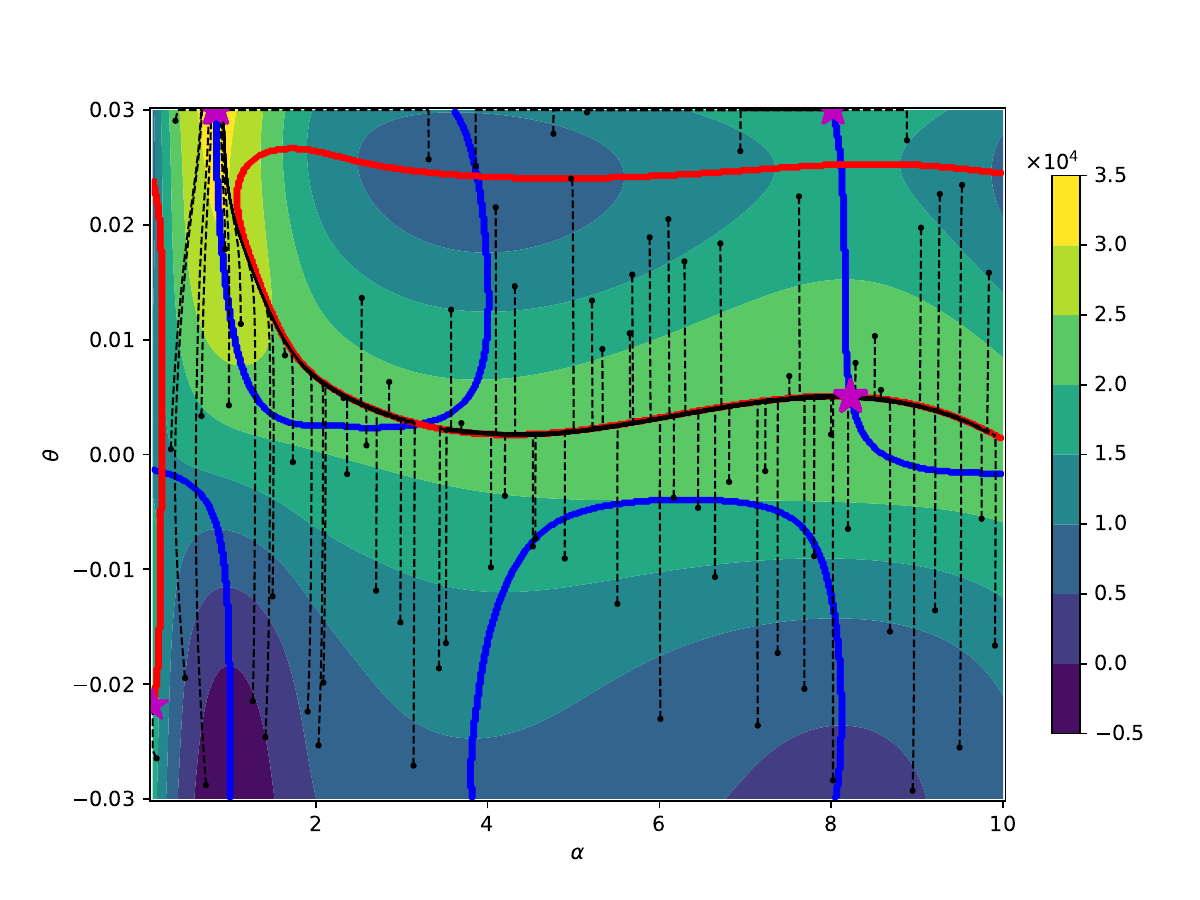} }}\\
    \caption{Left: numerical trajectories of the min-max problem solution (starting point: black dot; final point: magenta `star') plotted on the loss surface. Right: trajectories plotted on the loss contour plot, with the sets $\nabla_\alpha \L(\w^*(\lambda);\lambda)=0$ (in blue) and $\nabla_\theta \L(\w^*(\lambda);\lambda)=0$ (in red).}
    \label{fig: Trajectories}%
\end{figure}

However, in practice, the different magnitudes of $\nabla_\alpha\L$ and $\nabla_\theta\L$ cannot be neglected (see \Cref{fig: Gradient}). In fact, for $\alpha>2.65$ and $\theta<0.025$, the market prices of risk that satisfy $\nabla_\theta \L(\w^*(\lambda);\lambda)=0$ form a set of points with ``semi-robust'' hedge. This is given by the fact that every movement along the direction of $\theta$ leads to a significant decrease in the loss function, while movements along $\alpha$ might increase the loss, but the amount of the variation is negligible compared to the movements along $\theta$. Indeed, even though in such a region $\nabla_\alpha\L(\w^*(\lambda);\lambda)\neq 0$, its magnitude entails only restrained variations (see \Cref{fig: Trajectories}b). The intuition is that we can still build hedges that are more robust than the conditional hedge of the previous section, but the robustness is only against movements of $\theta$.

\section{Conclusion}
\label{sec: Conclusion}

We have proposed a framework for the valuation of the prepayment option embedded in mortgages (EPO). We have introduced a stochastic risk factor that aims to capture behavioral uncertainty. Our model enables us to lower the cost of prepayment providing an advantage to both the financial institution issuing the mortgage and the client buying it. The second stochastic risk factor is non-tradable, hence multiple equivalent martingale measures exist. Particularly, we showed that a misspecification of the market price of risk may have a significant effect on the EPO valuation, and entails a complex hedging problem. We proposed a path-wise replication of the EPO exposure. We solved a conditional hedge, under the assumption of known risk-neutral dynamics. The advantages of including nonlinear instruments in the hedge have been investigated and quantified. We developed numerical experiments to show that the methodology is flexible and can be used to focus on different metrics of interest, such as the expected shortfall (of particular relevance for risk management). We provided a methodology for robust replication. The problem is formulated as a saddle-point problem. Its solutions are robust strategies against misspecification of the market price of risk in the sense that they provide hedging strategies that guarantee a bounded loss in case of changes in the market price of risk.

\section*{Acknowledgments}
L. P. wants to thank the Rabobank ``squads'' \emph{Treasury Analytic and Innovation} and \emph{Prepayment} for their help in developing this research. Special thanks go to Dr. Francois Nissen and Prof. Antoon Pelsser for fruitful discussions and inspiration, and to Mr. Pedro Iraburu for the tireless support provided when dealing with the prepayment data.

\bibliography{Bibliography}

\appendix

\section{Proofs and lemmas}
In this appendix, the proofs of the results are presented.

\subsection{Proof of \texorpdfstring{\Cref{thm: DecreasingPrepaymentPrice}}{}}
\label{app: ProofTheorem}
\begin{proof}[Proof of \Cref{thm: DecreasingPrepaymentPrice}]
\label{proof: Theorem}
Notice that items \ref{enum: InterestOnly} and \ref{enum: NotionalAvailability} ensure that it is not possible to run out of notional prior to the end of the contract. Item \ref{enum: RationalIncentive} forces the client behavior to be digital: either no prepayment occurs or the maximum allowed quantity $u$ is prepaid. From item \ref{enum: PaymentsAtReset}, the payoff in \Cref{def: PrepaymentOptionPayoff} reads:
\begin{equation*}
    H(t_j)=\big(K - F(t_{j-1};t_{j-1},t_j)\big) N(t_{j-1}) \Delta t_j, \qquad j=1,\dots,n,
\end{equation*}
with the stochastic notional given by:
\begin{equation*}
    N(t_{j})=u \sum_{0\leq k\leq j}\1_{\Omega_K}(\kappa(t_k),b(t_k)),\qquad j=0,\dots,n-1,
\end{equation*}
where $\Omega_K=\big\{(x,y)\in\R^2: x+y\leq K\big\}$.

At $t=t_0$, the value function in \eqref{eq: PrepaymentValue} reads:
\begin{equation*}
    \begin{aligned}
        V(t_0)&=\E_{t_0}^\Q\bigg[\sum_{j\geq 1}\frac{M(t_0)}{M(t_j)}\big(K - F(t_{j-1};t_{j-1},t_j)\big)N(t_{j-1}) \Delta t_j\bigg]\\
        &=u\sum_{0\leq k\leq n-1}\E_{t_0}^\Q\bigg[\1_{\Omega_K}(\kappa(t_k),b(t_k))\sum_{k+1\leq j\leq n}\frac{M(t_0)}{M(t_j)}\big(K - F(t_{j-1};t_{j-1},t_j)\big) \Delta t_j\bigg]\\
        &=u\sum_{0\leq k\leq n-1}\E_{t_0}^\Q\bigg[\1_{\Omega_K}(\kappa(t_k),b(t_k))\frac{M(t_0)}{M(t_{k})}A(t_k)\big(K - \kappa(t_k)\big)\bigg].
    \end{aligned}
\end{equation*}
By partitioning $\Omega_K$, for any choice of $k=0,\dots,n-1$, we get:
\begin{equation*}
    \begin{aligned}
        \E_{t_0}^\Q\bigg[\1_{\Omega_K}(\kappa(t_k),b(t_k))\frac{M(t_0)}{M(t_{k})}&A(t_k)\big(K - \kappa(t_k)\big)\bigg]=\\
        &=\E_{t_0}^\Q\bigg[\1_{\Omega_K}(\kappa(t_k),b(t_k))\frac{M(t_0)}{M(t_{k})}A(t_k)\big[K - \kappa(t_k)\big]^+\bigg]\\
        &\qquad-\E_{t_0}^\Q\bigg[\1_{\Omega_K}(\kappa(t_k),b(t_k))\frac{M(t_0)}{M(t_{k})}A(t_k)\big[K - \kappa(t_k)\big]^-\bigg]\\
        &=P(t_0;t_k)\Bigg\{\E_{t_0}^{\Q_{t_k}}\bigg[\1_{\Omega_K}(\kappa(t_k),b(t_k))A(t_k)\big[K - \kappa(t_k)\big]^+\bigg]\\
        &\qquad\qquad\qquad-\E_{t_0}^{\Q_{t_k}}\bigg[\1_{\Omega_K}(\kappa(t_k),b(t_k))A(t_k)\big[K - \kappa(t_k)\big]^-\bigg]\Bigg\},
    \end{aligned}
\end{equation*}
where the last equality is obtained by changing the reference measure to the ${t_k}$-forward measure associated with the num\'eraire $P(t;t_k)$.

When we consider different volatilities $\eta_b$, we can apply \Cref{lem: Lemma} in the previous equation, for every $k=0,\dots,n-1$, setting $Y_\sigma=b(t_k;\eta_b)$, $X=\kappa(t_k)$ (with the probabilistic law given by the suitable ${t_k}$-forward dynamics), and $a=K$. In fact, thanks to item \ref{enum: IndependentB}, $b(t_k;\eta_b)$ is a normal random variable with mean 0 and variance  $\sigma^2=\frac{\eta_b^2}{2\alpha_b^\Q}(1-\e^{-2\alpha_b^\Q (t_k-t_0)})$, independent of $\kappa(t_k)$, and:
\begin{equation*}
    g_{1,2}(\kappa(t_k))=P(t_0;t_k)A(t_k)\big[K - \kappa(t_k)\big]^\pm
\end{equation*}
are deterministic positive functions of $\kappa(t)$ in the domains of interest (i.e., $\kappa(t_k)<K$ and $\kappa(t_k)>K$, respectively).
\end{proof}

\begin{lem}
\label{lem: Lemma}
    Consider real, independent random variables $X$ and $Y_\sigma$, for some parameter $\sigma\geq 0$. With $F_Y(y;\sigma)=\P[Y_\sigma\leq y]$ the CDF of $Y_{\sigma}$, we assume that, for any $\sigma_1<\sigma_2$, it holds that $F_Y(y;\sigma_1)>F_Y(y;\sigma_2)$ for any $y>0$. For $a>0$, we define the domains:
    \begin{equation*}
        \begin{aligned}
            A &= \Big\{(x,y)\in\R^2: x < a,\: y<a-x\Big\},\\
            B &= \Big\{(x,y)\in\R^2: x > a,\: y<a-x\Big\}.
        \end{aligned}
    \end{equation*}
    Given deterministic functions, $g_1(x)$ and $g_2(x)$, such that $g_1(x)>0$ for $x<a$ and $g_2(x)>0$ for $x>a$, then:
        \begin{align}
            \label{eq: ThesisITM}
            \E\big[g_1(X)\1_A(X,Y_{\sigma_1})\big]&>\E\big[g_1(X)\1_A(X,Y_{\sigma_2})\big],\\
            \label{eq: ThesisOTM}
            \E\big[g_2(X)\1_B(X,Y_{\sigma_1})\big]&<\E\big[g_2(X)\1_B(X,Y_{\sigma_2})\big],
        \end{align}
    for any $\sigma_1<\sigma_2$.
\end{lem}
\begin{proof}
By the definition of the expectation operator, and using the independence of $X$ and $Y_\sigma$, we get:
    \begin{equation*}
        \begin{aligned}
            \E\big[g_A(X)\1_A(X,Y_{\sigma})\big]&=\int_{-\infty}^{a}g_A(x)f_X(x)\int_{-\infty}^{a-x} f_Y(y;\sigma) \d y \d x,\\
            &=\int_{-\infty}^{a}g_A(x)f_X(x) F_Y(a-x;\sigma)\d x,
        \end{aligned}
    \end{equation*}
    where $f_X$ and $f_Y$ are the PDFs of $X$ and $Y$, respectively.
    The inequality in \eqref{eq: ThesisITM} is proved observing that the inequality $F_Y(y;\sigma_1)>F_Y(y;\sigma_2)$ holds for any $y>0$. A similar argument holds for \eqref{eq: ThesisOTM}.
\end{proof}

\subsection{Proof of \texorpdfstring{\Cref{prop: RecursiveRelationEPO}}{}}
\label{app: ProofProposition}

\begin{proof}[Proof of \Cref{prop: RecursiveRelationEPO}]
    \Cref{eq: PricingFormula} follows immediately from the definition of $V(t)$ given in \Cref{eq: PrepaymentValue}. Indeed, according to \eqref{eq: FutureCashflowsValue}, $\FV(t)$ is the value at time $t$ of all future cash flows excluding the very next one. The next cash flow value, namely the cash flow occurring at time $\tn(t)$, is obtained from the definition of the EPO cash flows in \eqref{eq: PrepaymentPayoff}. The discounted integral of the notional in \eqref{eq: PrepaymentPayoff} is split into a $t$-measurable (``known'') and a $ t$-non-measurable (``unknown'') part. Hence, we have:
    \begin{equation*}
        \begin{aligned}
            V(t)&=\E^{\Q}_t\bigg[\sum_{t_j\geq\tn(t)}\frac{M(t)}{M(t_j)}\CF(t_j)\bigg]\\
            &=\big(K-F(\tp(t);\tp(t),\tn(t))\big)\,\E^{\Q}_t\bigg[\frac{M(t)}{M(\tn(t))}\int_{\tp(t)}^{\tn(t)} N(\tau) \d\tau\bigg] + \E^{\Q}_t\bigg[\sum_{t_j>\tn(t)}\frac{M(t)}{M(t_j)}H(t_j)\bigg]\\
            &=\big(K-F(\tp(t);\tp(t),\tn(t))\big)\,\big(\CFk(t) + \CFu(t)\big)+\FV(t).
        \end{aligned}
    \end{equation*}
    
    Regarding the recursion given in \Cref{eq: RecursionH,eq: RecursionF,eq: RecursionHTrivial,eq: RecursionV}, \eqref{eq: RecursionHTrivial} is trivially obtained by definition. $\FV(t)$ in \eqref{eq: RecursionF} can be written in terms of a ``future'' $\FV(t_+)$ using the \emph{tower property} of conditional expectations, as long as the ``future cash flows'' are the same for $t$ and $t_+$, i.e. when $\tn(t)=\tn(t_+)$. 
    
    The case in \eqref{eq: RecursionV}, where $t$ has ``one more future cash flow'' than $t_+$, requires to include also its additional value, namely:
    \begin{equation*}
        \begin{aligned}
            \FV(t) &= \E^\Q\bigg[\frac{M(t)}{M(t_+)}\FV(t_+)+\frac{M(t)}{M(\tn(t_+))}\CF(\tn(t_+))\bigg]\\
            &=\E_t^\Q\bigg[\frac{M(t)}{M(t_+)}\E_{t_+}^\Q\bigg[\sum_{t_j>\tn(t_+)}\frac{M(t_+)}{M(t_j)}\CF(t_j)\bigg]+\frac{M(t)}{M(\tn(t_+))}\CF(\tn(t_+))\bigg]\\
            &=\E_t^\Q\bigg[\frac{M(t)}{M(t_+)}\E_{t_+}^\Q\bigg[\sum_{t_j\geq\tn(t_+)}\frac{M(t_+)}{M(t_j)}\CF(t_j)\bigg]\bigg]\\
            &=\E_t^\Q\bigg[\frac{M(t)}{M(t_+)}V(t_+)\bigg].
        \end{aligned}
    \end{equation*}

    The only remaining case, in \eqref{eq: RecursionH}, is obtained using the additive property of the integral. We have:
    \begin{equation*}
        \begin{aligned}
            \CFu(t)&=\E^\Q_t\bigg[\frac{M(t)}{M(\tn(t))}\int_{t}^{\tn(t)} N (\tau)\d \tau\bigg]\\
            &=\E^\Q_t\bigg[\frac{M(t)}{M(\tn(t))}\bigg(\int_{t}^{t_+} N (\tau)\d \tau+ \int_{t_+}^{\tn(t)} N (\tau)\d \tau\bigg)\bigg]\\
            &=\E_t^\Q\bigg[\frac{M(t)}{M(t_+)}P(t_+;\tn(t))\int_{t}^{t_+} N(\tau)\d\tau\bigg]+\E_t^\Q\bigg[\frac{M(t)}{M(t_+)}\E_{t_+}^{\Q}\bigg[\frac{M(t_+)}{M(\tn(t))}\int_{t_+}^{\tn(t)} N(\tau)\d\tau\bigg]\bigg],
        \end{aligned}
    \end{equation*}
    where \eqref{eq: RecursionH} follows since $\tn(t)=\tn(t_+)$.
\end{proof}

\section{EPO pricing algorithm}
\label{app: Algorithm}
In this appendix, we present the pseudo-code for the EPO pricing algorithm based on the results of \Cref{ssec: NumericalPricing}.

\begin{algorithm}[H]
\DontPrintSemicolon  
     \KwInput{Reset dates $\T_{\tt{r}}=\{t_0,t_1,\dots,t_{n-1}\}$, payment dates $\T_{\tt{p}}=\{t_1,t_2,\dots,t_{n}\}$ and time grid $\T_{grid}=\{\tau_0,\tau_1,\dots,\tau_{N_{steps}}\}$ such that $\T_{\tt{r}}\cup\T_{\tt{p}}\subset \T_{grid}$ and $\tau_0=t_0$ 
     State process $X_{i,k}:=[r_{i,k},b_{i,k},N_{i,k}]^\top$ for $i=1,\dots,N_{paths}$, $k=0,\dots,N_{steps}$}
     
     \KwOutput{Value process $V_{i,k}$ for $i=1,\dots,N_{paths}$ and $k=0,\dots,N_{steps}$}
    
     \tcp{Precompute money savings account $M(t)$, ZCB value $P(t;\tn(t))$, cash flow rate $K-L(\tp(t);\tp(t),\tn(t))$, and integrals $\int_{\tp(t)}^tN (\tau)\d \tau$ and $\int_{t}^{t_+} N(\tau)\d\tau$}
     $M_{i,k} := {\exp} \{{\tt{trapezoid}}_{t_0}^{\tau_k}(r_{i,\cdot},\tau_{\cdot})\}$ for $i=1,\dots,N_{paths}$ and $k=0,\dots,N_{steps}$\\
     $P_{i,k} := P(\tau_k;\tn(\tau_k))_i$ for $i=1,\dots,N_{paths}$ and $k=0,\dots,N_{steps}$\\
     $R_{i,k}:=K-L(\tp(\tau_k);\tp(\tau_k),\tn(\tau_k))_i$ for $i=1,\dots,N_{paths}$ and $k=0,\dots,N_{steps}$\\
     $\Ik_{i,k}:={\tt{trapezoid}}_{\tp(\tau_k)}^{\tau_k}(N_{i,\cdot}, \tau_{\cdot})$ for $i=1,\dots,N_{paths}$ and $k=0,\dots,N_{steps}$\\
     $\Iu_{i,k}:={\tt{trapezoid}}_{\tau_{k-1}}^{\tau_{k}}(N_{i,\cdot}, \tau_{\cdot})$ for $i=1,\dots,N_{paths}$ and $k=1,\dots,N_{steps}$\\
    
    \tcp{Initialise the EPO final value}
    $V_{i,N_{steps}}=0$ for $i=1,\dots,N_{paths}$
    
    \tcp{Initialise $\CFu$ and $\FV$}
    $\CFu_{i}={\FV}_{i}=0$ for $i=1,\dots,N_{paths}$
    
   \tcp{Backward induction}
    \For{$k=N_{steps}-1,\dots,0$}{
        \tcp{Fit regressors according with that cases in \Cref{prop: RecursiveRelationEPO}}
        \If{$\tau_k\in\T_{\tt{r}}$} {
            \For{$i=1,\dots,N_{paths}$} {
                $Y^{\FV}_i=\frac{M_{i,k}}{M_{i,k+1}}V_{i,k+1}$\tcp*{\Cref{eq: RecursionV}}
                }
                $\widehat{h_k}\equiv 0$\\
                $\widehat{f_k}={\tt{fit}}(X_{\cdot,k},Y_{\cdot}^{\FV})$}  
        \Else{
            \For{$i=1,\dots,N_{paths}$} {
                $Y^{\CF}_i=\frac{M_{i,k}}{M_{i,k+1}}\Big(P_{i,k+1} \Iu_{i,k+1} + \CFu_i\Big)$\tcp*{\Cref{eq: RecursionH}}
                $Y^{\FV}_i=\frac{M_{i,k}}{M_{i,k+1}}\FV_{i}$\tcp*{\Cref{eq: RecursionF}}
                }
                $\widehat{h_k}={\tt{fit}}(X_{\cdot,k},Y_{\cdot}^{\CF})$ \\
                $\widehat{f_k}={\tt{fit}}(X_{\cdot,k},Y_{\cdot}^{\FV})$
        }         
        \For{$i=1,\dots,N_{paths}$} {
        \tcp{Update value process $V_{i,k}$}
        \If{$\tau_k>t_n$} {$V_{i,k}=0$ \tcp*{No outstanding payments}}
        \Else {
        \tcp{Compute $\CFk$}
        $\CFk_i=P_{i,k}\Ik_{i,k}$\tcp*{\Cref{eq: IntegralKnown}}
        \tcp{Update $\CFu$ and $\FV$}
        $\CFu_i\approx\widehat{h_{k}}(X_{i,k})$\tcp*{\Cref{eq: TruncatedDoobH}}
        $\FV_i\approx\widehat{f_k}(X_{i,k})$\tcp*{\Cref{eq: TruncatedDoobF}}
        $V_{i,k}=R_{i,k} \cdot (\CFk_i + \CFu_i) + \FV_i$
        \tcp*{\Cref{eq: PricingFormula}}}     
        }
    }
\caption{EPO pricing with LSM.}
\label{alg: Algorithm}
\end{algorithm}

\end{document}